\documentclass[11pt]{article}
\usepackage{hyperref,amssymb,amsmath,amsfonts,color,amsthm}
\usepackage[all]{xy}
\usepackage[T1]{fontenc}
\usepackage[active]{srcltx}
\usepackage{hyperref}

\font\bigblack=msbm10 scaled\magstep 2 
\font\bbigblack=msbm10 scaled\magstep3

\voffset = - 1.2cm
\def\bigfield #1{\hbox{{\bigblack #1}}}
\def\bbigfield #1{\hbox{{\bbigblack #1}}}

\def\v #1{\vert #1\vert}             

\def\m #1 #2{(-1)^{{\v #1} {\v #2}}} 

\def\<#1>{\langle#1\rangle}        
\def\>#1{{\bf #1}}                
\def\f(#1,#2){\frac{#1}{#2}}

\def\dt2#1{\frac{d^2 #1}{dt^2}}

\def\ea{\varepsilon_a}

\def\G{{\Gamma}}

\def\big R{{\hbox{{\bigfield R}}}}
\def\bbig R{{\hbox{{\bbigfield R}}}}

\def\dim{\hbox{{\rm dim}}}        

\def\id{{\hbox{id}}}                    

\def\ba{\begin{eqnarray}}
	\def\ea{\end{eqnarray}}
\def\be{\begin{equation}}
	\def\ee{\end{equation}}

\def\<#1>{\langle#1\rangle}

\def\di{\bigstar}
\newcommand{\bea}{\begin{eqnarray}}
	\newcommand{\eea}{\end{eqnarray}}

\def\frac#1#2{{#1\over#2}} 

\theoremstyle{plain}
\newtheorem{theorem}{Theorem}
\newtheorem{corollary}[theorem]{Corollary}
\newtheorem{proposition}[theorem]{Proposition}
\newtheorem{lemma}[theorem]{Lemma}

\theoremstyle{definition}
\newtheorem{definition}[theorem]{Definition}
\newtheorem{note}[theorem]{Note}

\def\Hil{{\cal H}}

\def\<#1>{\langle#1\rangle}

\numberwithin{equation}{section}
\numberwithin{theorem}{section}

\begin{document}
	\nocite{*}
	
	\centerline{\Large {\bf Quantum quasi-Lie systems: properties and applications}} \vskip 0.75cm
	
	\centerline{J.F. Cari\~nena$^{\dagger}$, J. de
		Lucas$^{\ddagger}$ and C. Sard\'on$^{*}$} \vskip 0.5cm
	
	\centerline{$^{\dagger}$Department of Theoretical Physics and IUMA, University of Zaragoza,}
	\medskip
	\centerline{c. Pedro Cerbuna 12, 50009, Zaragoza, Spain.}
	\medskip
	\centerline{$^{\ddagger}$Department of Mathematical Methods in Physics, University of Warsaw,}
	\medskip
	\centerline{ul. Pasteura 5, 02-093, Warsaw, Poland.}
	\medskip
	\centerline{$^{*}$ Department of Applied Mathematics, Politechnic University of Madrid (UPM)}
	\medskip
	\centerline{c. José Gutiérrez Abascal 2, 28006, Madrid.}
	\medskip
	
	
	\vskip 1cm
	
	\begin{abstract}
		A {\it Lie system} is a non-autonomous system of ordinary differential equations describing the integral curves of a $t$-dependent vector field taking values in a finite-dimensional Lie algebra of vector fields. Lie systems have been generalised in the literature to deal with $t$-dependent Schr\"odinger equations determined by a particular class of $t$-dependent Hamiltonian operators, the {\it quantum Lie systems}, and other differential equations through the so-called {\it quasi-Lie schemes}. This work extends quasi-Lie schemes and quantum Lie systems to cope with $t$-dependent Schr\"odinger equations associated with the here called {\it quantum quasi-Lie systems}. To illustrate our methods, we propose and study a quantum analogue of the classical nonlinear oscillator searched by Perelomov and we analyse a quantum one-dimensional fluid in a trapping potential along with quantum $t$-dependent {Smorodinsky--Winternitz} oscillators.
	\end{abstract}
	
	
	
	\section{Introduction}
	A {\it Lie system} is a non-autonomous system of first-order ordinary differential equations whose general solution can be written via a (generally nonlinear) autonomous function, referred to as {\it superposition rule}, a finite set of particular solutions and constants related to the initial conditions \cite{CGM07,Dissertationes,PW}. For instance, non-autonomous {inhomogeneous}  linear systems of first-order ordinary differential equations, Bernoulli equations, Riccati equations, and matrix Riccati equations are examples of Lie systems \cite{LS21,PW}.

	Lie systems are mathematically interesting due to their geometric and algebraic properties \cite{CGL10,CGM07,LaforWinter}. The {\it Lie--Scheffers theorem} establishes that a Lie system is equivalent to a non-autonomous vector field taking values in a finite-dimensional Lie algebra  {$V$} of vector  {fields,} known {as} a {\it Vessiot--Guldberg Lie algebra} (VG Lie algebra henceforth) of the Lie system \cite{CGM07,Dissertationes}. As a consequence,  the study of Lie systems leads to the investigation of superposition rules through projective foliations on an {appropriate} bundle \cite{CGM07}, generalised distributions \cite{CGM07}, Lie group actions \cite{LaforWinter,PW}, etc. Furthermore, there has been a recent interest in Lie systems possessing VG Lie algebras of Hamiltonian vector fields regarding different geometric structures \cite{BBCHLS,BHLS,CCJL19,CGLS,CLS,GLMV19,HLS,LS21}. 
	
	The Lie--Scheffers theorem implies that being a Lie system is a very restrictive condition \cite{CGL10}. Nonetheless, Lie systems play a relevant r\^ole in mathematics, physics, and control theory, e.g. \cite{Dissertationes} cites more than 200 works on Lie systems and their applications. In particular, the theory of Lie systems provides methods to determine the integrability of certain systems of first-order differential equations \cite{AvCar,CGR01,Dissertationes,CLR07c}. Furthermore, Lie systems have been proved  {to be} very helpful in the study of geometric phases \cite{FloresLucas}, the solution of important nonlinear oscillators \cite{Dissertationes}, the analysis of Wei-Norman equations \cite{CLR07c,CMN12,ChandraSL,CK15}, and the research on problems occurring in quantum mechanics \cite{CCJL19,CLR07c} and biology \cite{BBCHLS,BHLS}. Riccati equations and their generalisations are ubiquitous in physics, occurring for instance in cosmology and financial models \cite{AvCar,BBCHLS,Dissertationes,Levin,sorwin}. Riccati equations appear also in the study of epidemiological models \cite{EFSZ21}. 
	
	A generalisation of the concept of Lie systems has recently been proposed in \cite{CGL08,Emden,Dissertationes}, where the theory of {\it quasi-Lie systems} was developed. This theory investigates non-autonomous systems of first-order differential equations  described by $t$-dependent vector fields, $X$, taking values in a finite-dimensional  {linear} space of vector fields $V$ containing a VG Lie algebra $W\subset V$ such that $[W,V]\subset V$. If $V$ is also a Lie algebra, then $X$ is a Lie system. The Lie algebra $W$ can be integrated to give rise to a local Lie group action {of its associated Lie group}, which can be employed to transform the original system into a Lie system with a VG Lie algebra within $V$ (cf. \cite{CGL08}). This idea allows one to analyse the integrability of non-harmonic  {classical} oscillators, dissipative Milne--Pinney equations, and other systems that cannot be studied through Lie systems \cite{CGL08,Emden,CL08Diss,Dissertationes,CL15}. 
	
	Besides, the theory of Lie systems has been extended to the quantum realm in   { \cite{CCR,CLquantum,Dissertationes,CLR07c,CarRam03}}. A {\it quantum Lie system} is a $t$-dependent Hamiltonian operator $\widehat{H}(t)$ such that ${\rm i}\widehat{H}(t)$ takes values in a finite-dimensional real Lie algebra of skew-Hermitian operators \cite{CCR,CLquantum,Dissertationes,CLR07c}. This allows us to develop powerful techniques to study $t$-dependent Schr\"odinger equations associated with these $t$-dependent Hamiltonian operators \cite{CCJL19,CCR,CLquantum,Dissertationes,CLR07c}. For instance, $t$-dependent Schr\"odinger equations associated with quantum Lie systems can be investigated through Lie systems on Lie groups and the corresponding real Lie algebra of skew-Hermitian operators is a clue to solve them  \cite{Dissertationes}. An example treatable with this aid is the case of quantum $t$-dependent dissipative harmonic oscillators of Calogero--Moser type \cite{Dissertationes}.
	
	To enlarge the field of applications of quantum Lie systems, we propose as a main result a quantum analogue of the theory of quasi-Lie schemes based on the hereupon denominated {\it quantum quasi-Lie schemes}. Every quantum quasi-Lie scheme gives rise to a group of $t$-dependent gauge transformations, the {\it group of the quantum quasi-Lie scheme}, which enables one to transform certain $t$-dependent Schr\"odigner equations into new ones that can be studied through quantum Lie systems. This stands out a particular family of $t$-dependent Schr\"odinger equations, described by the here defined {\it quantum quasi-Lie systems} (QQLS), which cover the $t$-dependent Schr\"odinger equations associated with quantum Lie systems as a particular case. In this manner, quantum quasi-Lie systems  can be investigated via the theory of quantum Lie systems. 
	
	One of the main uses of quantum quasi-Lie schemes is to map a $t$-dependent Hamiltonian operator $-{\rm i}\widehat H(t)$ described through a quantum quasi-Lie scheme into a quantum quasi-Lie system or a simpler $t$-dependent Hamiltonian operator $-{\rm i}\widehat H'(t)$ through an element $U(t)$ of the group of the quantum quasi-Lie scheme. Instead of applying straightforwardly a generic $U(t)$ to transform $-{\rm i}\widehat H(t)$ into a simpler $-{\rm i}\widehat{H}'(t)$, which is operationally complicated and tedious but ubiquitous in the literature, we provide a series of results, e.g. Propositions \ref{Red1} and \ref{Red2}, giving easy criteria to predict some interesting properties of $-{\rm i}\widehat H'(t)$. This will eventually  gives hints on the form of $U(t)$ to simplify $-{\rm i}\widehat H'(t)$ or to ensure that $-{\rm i}\widehat H(t)$ is a quantum quasi-Lie system.
	Mathematically, Propositions \ref{Red1} and \ref{Red2} concern the study of real non semi-simple Lie algebra representations. This is the least studied case in the literature, which mainly focuses upon representations of complex semi-simple Lie algebras. 
	
	As applications, we show that quantum quasi-Lie schemes allow us to study the quantum  {evolution} in problems in which the theory of quantum Lie systems does not apply. In particular, we develop a quantum analogue of the classical anharmonic oscillator formerly  studied by Perelomov \cite{Pe78}. This solves the problem proposed by Perelomov of finding a quantum analogue of his oscillator. 
	Subsequently, we focus on the application of quantum quasi-Lie systems in $t$-dependent Schr\"odinger equations describing a one-dimensional quantum fluid in a $t$-dependent trapping potential \cite{Su97} and a quantum $t$-dependent frequency Winternitz--Smorodynski oscillator. 
	
	The paper is organised  as follows. Section \ref{SLSQM} reviews fundamental notions of differential geometry on Hilbert spaces. Section \ref{Sec:QLS} presents quantum Lie systems and their properties. Section \ref{Sec:QQLS} introduces the theory of quantum quasi-Lie schemes, while the theory of quantum quasi-Lie systems is presented  in Section \ref{Sec:QQLSy}. Section \ref{Sec:TPQQLS} addresses the description of new methods to study $t$-dependent Schr\"odinger equations by means of quantum quasi-Lie schemes. Subsequently, we present an application of the theory of quantum quasi-Lie schemes and systems to quantum nonlinear oscillators with a homogeneous potential in Section \ref{Applications}. Section \ref{minos2}  is devoted to the application of the theory of quantum quasi-Lie schemes to the particular case of homogeneous potentials of degree minus two. In particular, this entails studying quantum $t$-dependent frequency Calogero--Moser systems, quantum fluids in a $t$-dependent homogeneous trapping potential, and $t$-dependent quantum  {Smorodinsky--Winternitz} systems. To conclude, Section \ref{Sec:Sum} contains a summary of the obtained results and a commentary on future research prospects.  
	
	\section{Differential geometry in Hilbert spaces.}\label{SLSQM}
	
	This section surveys the differential geometry of (possibly infinite-dimensional) Hilbert spaces. For further details on the geometry of infinite-dimensional manifolds and the theory of skew-Hermitian operators, we refer to  {\cite{GKM, Ha13,La62,Le13}}.
	
	Let $\mathcal{H}$ be a complex Hilbert space endowed with the topology inherited from the norm $\|\cdot\|$ associated with its scalar product $\langle\cdot |\cdot\rangle$. It is hereupon assumed that $\mathcal{H}$ is separable, as it usually happens in quantum mechanics. 
	The separability of $\mathcal{H}$ {means  that there exists a countable orthonormal basis}. The orthonormal basis, let us say $\{\psi_n\}_{n \in\mathbb{N}}$, gives rise to a global chart 
	$$
	\varphi:\psi\in \mathcal{H}\mapsto \{\mathfrak{ Re}\langle \psi_n|\psi\rangle,\mathfrak{Im}\langle \psi_n|\psi\rangle\}_{n\in\mathbb{N}}\in \ell^2(\mathbb{R})
	$$ 
	onto the real Hilbert space $\ell^2(\mathbb{\mathbb{R}})$ of square summable sequences, which allows us to consider $\mathcal{H}$ as a manifold $\mathcal{H}_{\mathbb{R}}$ modelled over a Hilbert space \cite{CCJL19,Dissertationes}. 
	
	We call {\it derivative} at a point $p$ in an open subset  $U\subset \mathcal{H}$ a linear mapping $D:C^\infty(U)\rightarrow \mathbb{R}$ satisfying that $D(fg)=(Df)g(p)+f(p)(Dg)$ for every $f,g\in C^\infty(U)$. 	{Every $\psi\in \mathcal{H}$ induces a derivative $\dot{\psi}_\phi:C^\infty(U)\rightarrow\mathbb{R}$ at the point $\phi$ in an open subset $U$ of $\mathcal{H}$ given by
		\begin{equation}\label{equat}
			(\dot{\psi}_\phi f)=\frac{d}{dt}\bigg|_{t=0}f(\phi+t\psi),\quad \forall f\in C^{\infty}(U).
	\end{equation}}
	
	{A derivative of the above particular form is called a {\it kinematic tangent vector} with foot point $\phi$. }  The term `kinematic' is used in the literature to distinguish kinematic tangent vectors from {\it operational tangent vectors}, which refer to derivatives on $C^\infty(U)$. There is, in general, no one-to-one relation between both types of tangent vectors on infinite-dimensional manifolds ({see e.g. \cite{BGT18,KM97})}. The term kinematic in kinematic tangent vectors and related notions will be hereafter skipped to simplify the presentation of our results. The {\it tangent space} $T_\phi\mathcal{H}_\mathbb{R}$ at $\phi\in \mathcal{H}_{\mathbb{R}}$ is the space  of all tangent vectors with foot point $\phi$ \cite{Ha13,La62,Le13}. It is immediate to prove, e.g. by inspection of (\ref{equat}) on smooth functions $f_{\widetilde{\psi}}: \widetilde{\phi}\in \mathcal{H}\mapsto \mathfrak{Re}\langle \widetilde{\psi}|\widetilde{\phi}\rangle$ for $\widetilde{\psi}\in \mathcal{H}$, that $\dot \psi_\phi\neq \dot \psi'_\phi$,  for $\psi,\psi',\phi\in \mathcal{H}$, if and only if $\psi\neq \psi'$. Hence, each $T_\phi\mathcal{H}_\mathbb{R}$ is isomorphic to $\mathcal{H}_\mathbb{R}$ for every $\phi\in \mathcal{H}_\mathbb{R}$. The isomorphism associates $\psi\in \mathcal{H}_\mathbb{R}$ with the  tangent vector $\dot{\psi}_\phi\in T_\phi\mathcal{H}_\mathbb{R}$. To simplify the notation, $\dot{\psi}_\phi$ will be denoted by $\dot{\psi}$ if the foot point is known from the context. The space $T_\phi\mathcal{H}_\mathbb{R}$ admits a unique norm turning the isomorphism with $\mathcal{H}_\mathbb{R}$ into an isometry. Similarly, its {\it tangent bundle},  {$T\mathcal{H}_\mathbb{R}:=\bigsqcup_{\phi\in \mathcal{H}} T_\phi\mathcal{H}_\mathbb{R}$}, where $\bigsqcup$ stands for a {\it disjoint union} \cite{Le11,Le13}, is naturally diffeomorphic to $\mathcal{H}_\mathbb{R}\oplus\mathcal{H}_\mathbb{R}$. The space $T\mathcal{H}_\mathbb{R}$ is indeed a vector bundle relative to $\pi:\dot \psi_\phi\in T\mathcal{H}_\mathbb{R}\mapsto \phi\in \mathcal{H}_\mathbb{R}$.  
	
	A {\it  vector field} is a mapping $X:D(X)\subset \mathcal{H}_\mathbb{R}\rightarrow T\mathcal{H}_\mathbb{R}$ from a dense subspace $D(X)\subset\mathcal{H}_\mathbb{R}$ satisfying that $\pi\circ X={\id}_{\mathcal{H}_\mathbb{R}}|_{D(X)}$. For any vector field $X$ on $\mathcal{H}_\mathbb{R}$ and $\phi\in D(X)$, we  denote $X(\phi):=(\phi,X_\phi)\in T_\phi\mathcal{H}_\mathbb{R}\subset T\mathcal{H}_\mathbb{R}\simeq   \mathcal{H}_\mathbb{R}\oplus \mathcal{H}_\mathbb{R}$. 
	
	A vector field $X$ admits
	a {\it flow} if there exists a continuous map $Fl^X:(t,\phi)\in \mathcal{U}\subset \mathbb{R}\times\mathcal{H}_\mathbb{R}\mapsto Fl_t^X(\phi)\in \mathcal{H}_\mathbb{R}$ such that $\{0\}\times \mathcal{H}_{\mathbb{R}}\subset \mathcal{U}$, $Fl^X_0(\phi)=\phi$ for every $\phi\in \mathcal{H}_\mathbb{R}$ and
	$$
	\frac{d}{dt}Fl^{X}(t,\phi)=\lim_{s\rightarrow 0}\frac{Fl^X(t+s,\phi)-Fl^X(t,\phi )}{s}=X(Fl^{X}(t,\phi)),
	$$
	where the limit is relative to the norm topology of $\mathcal{H}_\mathbb{R}$ \cite{Ha13,Sc13}.	
	
	The main objects to be considered hereafter are the  vector fields of the form 
	$$X_{\widehat{A}}:\psi\in D(\widehat{A})\subset \mathcal{H}_{\mathbb{R}}\mapsto \dot{(\widehat{A}\psi)}_{\psi}\in T_\psi\mathcal{H}_{\mathbb{R}}\subset T\mathcal{H}_\mathbb{R}
	$$ 
	for a skew-Hermitian operator $\widehat{A}:D(\widehat{A})\subset \mathcal{H}\rightarrow \mathcal{A}$, where $D(\widehat{A})$ is the domain of the operator $\widehat{A}$. It is worth noting that an operator on an infinite-dimensional Hilbert space only needs to be defined on a dense subspace of the Hilbert space \cite{Ha15}.	Even if $\widehat A$ is not bounded, which implies in virtue of the {\it Hellinger--Toeplitz Theorem}  \cite{Ha13} that $X_{\widehat{A}}$ can only be defined on a dense subset of $\mathcal{H}_\mathbb{R}$, the {\it Stone's Theorem} ensures 
	that $\widehat{A}$ gives rise to a strongly continuous one-parameter group $\{e^{t\widehat{A}}\}_{t\in\mathbb{R}}$ of operators \cite{Stone,Stone2}. 
	
	Previous results ensure that the integral curves of $X_{\widehat{A}}$, namely $t\mapsto e^{t\widehat{A}}\phi$ for every $\phi\in \mathcal{H}$, are defined on the whole $\mathcal{H}_\mathbb{R}$ and $X_{\widehat{A}}$ is complete. This fact can be used to consider skew-Hermitian operators as fundamental vector fields relative to the standard action of the unitary group $U(\mathcal{H})$ on the Hilbert space $\mathcal{H}$, relative to the strong topology. If  $\widehat{A}:=-{\rm i}\widehat H$ for  an autonomous self-Hermitian Hamiltonian $\widehat H$, then the solutions to the associated $t$-dependent Schr\"odinger equation describe the integral curves of $X_{-{\rm i}\widehat{H}}$ \cite{BCG,CLR07c,Ha01}. 

	Assume that $\widehat{B}_1$ and $\widehat{B}_2$ are two, possibly  not bounded, operators acting on $\mathcal{H}$ admitting a common dense domain $D$ that is also invariant under $\widehat{B}_1$ and $\widehat{B}_2$, namely $\widehat{B}_1,\widehat{B}_2$ are defined on $D$ and $\widehat{B}_1(D),\widehat{B}_2(D)\subset D$. The commutator of $\widehat{B}_1$ and $\widehat{B}_2$ is defined to be $[\widehat{B}_1,\widehat{B}_2]\phi=(\widehat{B}_1\widehat{B}_2-\widehat{B}_2\widehat{B}_1)\phi$ for every $\phi\in D$.

	For two vector fields  $X_{\widehat{B}_1},X_{\widehat{B}_2}\in \mathfrak{X}(\mathcal{H}_\mathbb{R})$  satisfying that $\widehat{B}_1,\widehat{B}_2$ are skew-Hermitian operators having a common invariant dense domain, we define 
	$$
	[X_{\widehat{B}_1},X_{\widehat{B}_2}]:=-X_{[\widehat{B}_1,\widehat{B}_2]}.
	$$
	
	The above expression, when the integral curves of $\widehat{B}_1,\widehat{B}_2$ are smooth enough, arises as a consequence of the definition of the geometric Lie bracket of vector fields on a Banach manifold. Let us analyse this fact in detail.	Assume that $X_{\widehat{B}_1},X_{\widehat{B}_2}$ are such that $\widehat{B}_1,\widehat{B}_2$ admit a common dense invariant domain. If the curve $
	\alpha(t):=\exp(-t\widehat{B}_1)\exp(-t\widehat{B}_2)\exp(t\widehat{B}_1)\exp(t\widehat{B}_2)\phi
	$	is smooth at $\phi \in\mathcal{H}$, it is possible to retrieve the Lie bracket of the vector fields $X_{\widehat{B}_1},X_{\widehat{B}_2}$ geometrically \cite[Theorem 1.33]{Olver}:
	\begin{equation*}
		[[X_{\widehat{B}_1},X_{\widehat{B}_2}]](\phi)=\dfrac 12 \dfrac{\partial^2}{\partial
			t^2}\bigg|_{t=0}\left(Fl^{X_{\widehat{B}_2}}_{-t}\circ Fl^{X_{\widehat{B}_1}}_{-t}\circ Fl^{X_{\widehat{B}_2}}_t\circ
		F^{X_{\widehat{B}_1}}_t(\phi)\right)\!=\!\dfrac 12 \dfrac{\partial^2}{\partial
			t^2}\bigg|_{t=0}e^{-t\widehat{B}_2}e^{-t\widehat{B}_1}e^{t\widehat{B}_2}e^{t\widehat{B}_1}\phi.
	\end{equation*}	
	In particular, the above always holds when $\mathcal{H}$ is finite-dimensional.
	A {\it $t$-dependent vector field} on  $\mathcal{H}_\mathbb{R}$ is a mapping $X:\mathcal{U}\subset \mathbb{R}\times \mathcal{H}_\mathbb{R}\rightarrow T\mathcal{H}_\mathbb{R}$, for a certain open subset  $\mathcal{U}$ in $\mathbb{R}\times \mathcal{H}_{\mathbb{R}}$, such that $\pi\circ X(t,\psi)=\psi$ on ${\mathcal{U}}$ and the domain of each $X_t:\phi\in \mathcal{U}\cap (\{t\}\times \mathcal{H}_\mathbb{R})\mapsto X(t,\phi)\in T\mathcal{H}_\mathbb{R}$, with $t\in \mathbb{R}$, is dense in $ \mathcal{H}_\mathbb{R}$. The {\it integral curves} of $X$ are sections $\gamma:\mathbb{R}\rightarrow \mathbb{R}\times \mathcal{H}_\mathbb{R}$ of the bundle $\pi_\mathbb{R}:\mathbb{R}\times \mathcal{H}_\mathbb{R}\rightarrow \mathbb{R}$ that are simultaneously integral curves of the vector field $\partial_t+X$ on $\mathbb{R}\times \mathcal{H}_{\mathbb{R}}$. If $\pi_2:\mathbb{R}\times \mathcal{H}_\mathbb{R}\rightarrow \mathcal{H}_\mathbb{R}$ is the projection onto the second manifold, the curves $\gamma$ are given by  the solutions to the differential equation
	$$
	\frac{d {(\pi_2\circ\gamma)} }{d t}(t)=X\circ \gamma(t).
	$$
	
	Every $t$-dependent Hermitian operator $\widehat H(t)$ on $\mathcal{H}$ gives rise to a $t$-dependent vector field $X_{-{\rm i}\widehat H(t)}$ on $\mathcal{H}_{\mathbb{R}}$ and an associated $t$-dependent Schr\"odinger equation
	\begin{equation*}
		\frac{\partial \psi}{\partial t}=-{\rm i}\widehat{H}(t)\psi,\qquad \psi\in \mathcal{H},\qquad t\in \mathbb{R}.
	\end{equation*}

	\section{Quantum Lie systems}\label{Sec:QLS}
	
	Let us survey the quantum version of Lie systems \cite{Dissertationes,CLR07c}. 
	On a finite-dimensional manifold $N$, every finite-dimensional real Lie algebra $V$ of vector fields on $N$ can be integrated to give rise to a local Lie group action on $N$. In the infinite-dimensional analogue, additional technical conditions on $V$ may be necessary to integrate $V$. In particular, we are interested in the integration to a Lie group action of finite-dimensional real Lie algebras of vector fields induced by skew-Hermitian operators on (possibly infinite-dimensional) manifolds (see \cite{FSSS72,Ne59,Si72} for details). In all cases studied in this work, $V$ can be found to be integrable to a (at least local) Lie group action. This Lie group action allows us to study quantum Lie systems via  {the standard theory of}  Lie systems \cite{CarRam03,CarRam05b}. To focus on the applications of our ideas, quantum quasi-Lie schemes, and quantum Lie systems to be studied in next sections, further comments on the integrability of $V$ to a Lie group action will be omitted.  
	
	\begin{definition}
		A {\it quantum Lie system} is a $t$-dependent operator $\widehat{H}(t)$ on $\mathcal{H}$ of the form
		\begin{equation}\label{introham}
			\widehat H(t):=\sum_{\alpha=1}^rb_{\alpha}(t)\widehat{H}_{\alpha},
		\end{equation}where $b_1(t),\ldots,b_r(t)$ are	$t$-dependent  real functions and $\mathfrak{V}:=\langle {\rm i}\widehat{H}_1,\ldots,{\rm i}\widehat{H}_r\rangle$ is  {a 
			$r$-dimensional}  real Lie algebra of skew-Hermitian operators. We call $\mathfrak{V}$  {a} {\it Vessiot--Guldberg (VG) Lie algebra} for the quantum Lie system (\ref{introham}) \cite{CLR07c,CarRam03,CarRam05b}.
	\end{definition}

	Each quantum Lie system (\ref{introham}) determines a $t$-dependent Schr\"{o}dinger equation 
	\begin{equation}\label{schreq}
		\frac{\partial \psi}{\partial t}=-{\rm i}\widehat{H}(t)\psi=-\sum_{\alpha=1}^rb_{\alpha}(t){\rm i}\widehat{H}_\alpha\psi,\qquad \psi\in \mathcal{H},\qquad t\in \mathbb{R}.
	\end{equation}Its particular solutions are then integral curves of the $t$-dependent vector field $X_{-{\rm i}\widehat{H}(t)}$ (see \cite{CLR07c} for details).

	Let us illustrate quantum Lie systems through an example. Consider the quantum one-dimensional {\it Caldirola{--}Kanai oscillator} \cite{ Ba31, Ca48, Ikot,Ka48}, which is determined by the $t$-dependent operator
	\begin{equation}\label{CK}
		\widehat{H}_{ {\textrm{CK}}}(t):=\exp{(-2\gamma t)}\frac{\widehat{p}^2}{2}+\omega^2 \exp{(2\gamma t)}\frac{\widehat{x}^2}{2},\qquad \gamma\in \mathbb{R},\,\,\omega\in \mathbb{R},
	\end{equation}
	where {\it $\widehat{x}$} and $\widehat{p}$ refer to the  {usual} position and momentum operators on $\mathbb{R}$. This model is a quantum analogue of the classical harmonic oscillator with a $t$-dependent mass $m(t):= \exp{(2\gamma t)}$ and  {a} constant frequency $\omega$  {\cite{CLR07c,GV09}}. 
	
	The $t$-dependent operator $\widehat H_{ {\textrm{CK}}}(t)$ can be written as a linear combination with $t$-dependent real functions  of the Hermitian operators
	$$
	\widehat{H}_1:=\widehat{x}^2/2,\qquad \widehat{H}_2:=\widehat{p}^2/2,\qquad
	\widehat{H}_3:=(\widehat{x}\widehat{p}+\widehat{p}\widehat{x})/4.
	$$
	The {real linear} space  $\mathfrak{V}_{ {\textrm{CK}}}:=\langle {\rm i}H_1,{\rm i}H_2,{\rm i}H_3\rangle$ is a vector space of skew-Hermitian operators. Additionally,
	\begin{equation*}
		[{\rm i}\widehat{H}_1,{\rm i}\widehat{H}_2]=-2{\rm i}\widehat{H}_3,\quad [{\rm i}\widehat{H}_1,{\rm i}\widehat{H}_3]=-{\rm i}\widehat{H}_1,\quad [{\rm i}\widehat{H}_3,{\rm i}\widehat{H}_2]=-{\rm i}\widehat{H}_2.
	\end{equation*}
	Hence, $\mathfrak{V}_{ {\textrm{CK}}}$ is a finite-dimensional Lie algebra  {isomorphic to $\mathfrak{sl}(2, \mathbb{R})$}. Since ${\rm i}\widehat H_{{\textrm{CK}}}(t)$ takes values in $\mathfrak{V}_{ {\textrm{CK}}}$, the $t$-dependent operator $\widehat H_{ {\textrm{CK}}}(t)$ is a quantum Lie system.
	
	The  {solutions of the} $t$-dependent Schr\"odinger equation associated with (\ref{CK}), namely
	\begin{equation*}
		\frac{\partial \psi}{\partial t}=-{\rm i}\widehat H_{ {\textrm{CK}}}(t)\psi=-{\rm i}\left(\exp{(-2\gamma t)}\frac{\widehat{p}^2}{2}+\omega^2 \exp{(2\gamma t)}\frac{\widehat{x}^2}{2}\right)\psi,
	\end{equation*}
	{are} the integral curves of $X_{ {\textrm{CK}}}(t):=\omega^2 \exp{(2\gamma t)}X_{-{\rm i}\widehat{H}_1}+\exp{(-2\gamma t)}X_{-{\rm i}\widehat{H}_2}$.
	
	Similarly to standard Lie systems, the associated $t$-dependent Schr\"odinger equation related to a quantum Lie system can be solved by means of a Lie system on a Lie group \cite{CGM00,Dissertationes}. Let us explain this fact  (see \cite{Dissertationes} for applications).  Reconsider again the quantum Lie system (\ref{introham}), which satisfies that 
	\begin{equation*}
		[{\rm i}\widehat{H}_\alpha,{\rm i}\widehat{H}_\beta]=\sum_{\gamma=1}^rc_{ {\alpha\beta}}\,^\gamma\ {\rm i}\widehat{H}_\gamma,\qquad
		c_{{\alpha\beta}}\,^\gamma\in\mathbb{R},\qquad \alpha,\beta,\gamma=1,\ldots,r.
	\end{equation*}
	Let $\varphi:\mathfrak{g}\rightarrow \mathfrak{V}$ be an isomorphism of Lie algebras, where $\mathfrak{g}$ is an abstract Lie algebra, $G_\mathfrak{g}$ is the connected and simply connected Lie group related to $\mathfrak{g}$, and $\mathfrak{V}$ is the Lie algebra relative to the commutator of skew-Hermitian operators (more precisely defined on a common domain for all its elements given by a dense subspace $D_\mathfrak{V}\subset\mathcal{H}$). We define a continuous unitary action
	$\Phi:G_{\mathfrak{g}}\times\Hil\to\Hil$   (continuous relative to the norms induced by the one on $\mathcal{H}$ and the one in $G_{ \mathcal{\mathfrak{g}}}$), such that  	\begin{equation}\label{exp}
		\varphi(v)\psi=\frac{d}{dt}\bigg|_{t=0}\Phi({\exp(-tv)},\psi),\qquad \forall \psi\in D_\mathfrak{V},\quad \forall v\in\mathfrak{g}.
	\end{equation}
	Note that $\Phi$ amounts to a Lie group morphism $\widehat{\Phi}:g\in G\mapsto \Phi^g\in  U(\mathcal{H})$ where $\Phi^g(\psi):=\Phi(g,\psi)$ for every $g\in G_\mathfrak{g}$ and $\psi\in \mathcal{H}$, whilst $U(\mathcal{H})$ amounts to the space of unitary operators on $\mathcal{H}$. 
	We define $\Phi_\psi:g\in G_\mathfrak{g}\mapsto  \Phi(g,\psi)\in \mathcal{H}_\mathbb{R}$ for every $\psi\in \mathcal{H}_{\mathbb{R}}$. Note that $\varphi$ is the infinitesimal Lie algebra morphism induced by $\Phi$. Due to the relation $[X_{\widehat{B}_1},X_{\widehat{B}_2}]=-X_{[{\widehat{B}_1},{\widehat{B}_2}]}$, it follows from (\ref{exp}) that we can define $X_{-\varphi(v)}$ to be the fundamental vector field on $\mathcal{H}_\mathbb{R}$ corresponding to the element $v\in\mathfrak{g}$ and  then there exists   a Lie algebra isomorphism   $v\mapsto X_{ -\varphi(v)}$ between the elements of the Lie algebra $\mathfrak{g}$ and the Lie algebra of vector fields $V$ given by the vector fields on $\mathcal{H}_{\mathbb{R}}$ induced by the skew-Hermitian operators $\varphi(v)$ for $v\in \mathfrak{g}$. 
	
	Let $\{{\rm a}_1,\ldots, {\rm a}_r\}$ be the basis of  $T_eG\simeq \mathfrak{g}$ given by ${\rm a}_{\alpha}:=\varphi^{-1}({\rm i}\widehat{H}_\alpha)$ for $\alpha=1,\ldots,r$. Then, ${\rm a}_1,\ldots, {\rm a}_r$ have the same commutation relations, relative to the natural Lie bracket in $T_eG$ (see \cite{AM78}), as the operators ${\rm i}\widehat{H}_1,\ldots, {\rm i}\widehat{H}_r$, which in turn have the opposite structure constants of the vector fields $X_{{\rm i}\widehat{H}_\alpha}$. Hence, if $[{\rm i}\widehat{H}_\alpha,{\rm i}\widehat{H}_\beta]=\displaystyle{\sum_{\gamma=1}^r}c_{\alpha\beta}\,^\gamma\,   {\rm i}\widehat{H}_\gamma$, then $[X_{-{\rm i}\widehat{H}_\alpha},X_{-{\rm i}\widehat{H}_\beta}]=\displaystyle{\sum}_{\gamma=1}^rc_{\alpha\beta}\,^\gamma X_{-{\rm i}\widehat{H}_\gamma}$, where we assume $\alpha,\beta=1,\ldots,r$. Hence,
	\begin{equation*}
		[{\rm a}_\alpha,{\rm a}_\beta]=\sum_{\gamma=1}^rc_{\alpha\beta}\,^\gamma\,{\rm a}_\gamma ,\qquad
		c_{\alpha\beta}\, ^\gamma\in\mathbb{R}, \qquad \alpha, \beta, \gamma=1,\dots, r,
	\end{equation*}
	and
	\begin{equation*}
		\Phi(\exp(-\lambda {\rm a}_\alpha),\psi)=\exp(-{\rm i}\lambda \widehat{H}_\alpha)\psi,\qquad \forall \psi\in D_\mathfrak{V},\quad \forall \lambda\in \mathbb{R},\qquad \alpha=1,\ldots,r.
	\end{equation*}
	It is quite useful in applications to consider $G$ to be a matrix Lie group. In such a case,  the elements of $ \mathfrak{g}$ are matrices and the commutator of the elements of $T_eG$ is the commutator of the corresponding associated elements of the Lie algebra. It is also worth noting that $
	\varphi({\rm Ad}_g(v))={\rm Ad}_{\Phi^g}\varphi(v),$ for all $g\in G$ and $ v\in \mathfrak{g}
	$, where ${\rm Ad}_g$ stands for the adjoint action of $g$ on $\mathfrak{g}$ and ${\rm Ad}_{\Phi^g}$ is the adjoint action of $\Phi^g\in U(\mathcal{H})$ relative to $\varphi(v)$.
	
	
	In particular, let us prove that solving the $t$-dependent Schr\"odinger equation associated with a quantum Lie system  $\widehat H(t)={\displaystyle{\sum_{\alpha=1}^r}} b_\alpha(t) H_\alpha$ reduces to solving the Lie system in $G$ given by
	$$
	\frac{dg}{dt}=-\sum_{\alpha=1}^rb_\alpha(t)X_\alpha^R(g),
	$$ where $X_\alpha^R(g)=R_{g*e}{\rm a}_\alpha$ for $\alpha=1,\ldots,r$ and $R_g:g'\in G\mapsto g'g\in G$. Indeed, if we define $
	\psi(t)=\Phi(g(t),\psi(0))$, with   $\psi(0)\in \mathcal{H}_{\mathbb{R}}$, then for every $t_0\in \mathbb{R}$, one has that 
	\begin{multline*}
		\frac{d}{dt}\bigg|_{t=t_0}\Phi(g(t),\psi(t_0))=	\frac{d}{dt}\Phi(g(t)g(t_0)^{-1}g(t_0),\psi(t_0))\\=\varphi\left({\frac{dg}{dt}(t_0)g^{-1}}(t_0)\right)\Phi(g(t_0),\psi(0))=-{\rm i}\widehat H(t_0)\psi(t_0),
	\end{multline*}
	and $\psi(t)$ is the general solution to (\ref{schreq}).
	

	In the particular case of $\widehat H_{ {\textrm{CK}}}(t)$, which is related to a quantum VG Lie algebra isomorphic to $\mathfrak{sl}(2,\mathbb{R})$, the associated Lie system should be associated with the universal covering {group}, $\widetilde{SL}(2,\mathbb{R})$, of the Lie group $SL(2,\mathbb{R})$, which has a Lie algebra isomorphic to $\mathfrak{sl}(2 {,\mathbb{R}})$.  In consequence, it is defined on the Lie group $G:=\widetilde{SL}(2,\mathbb{R})$ and takes the form 
	\begin{equation}\label{equation}
		\frac{dg}{dt}=-e^{-2\gamma t}X_2^R(g)-\omega^2e^{2\gamma t}X_1^R(g),\qquad g\in \widetilde{SL}(2,\mathbb{R}),\qquad t\in \mathbb{R},
	\end{equation}
	where $X^R_\alpha(g):= {R_{g*e}}{\rm a}_\alpha$ for $\alpha=1,2,3$ and $\{{\rm a}_1,{\rm a}_2,{\rm a}_3\}$ is a basis of $T_{\rm Id}SL_2\simeq \mathfrak{sl}(2,\mathbb{R})$ satisfying the same commutation relations as the ${\rm i}\widehat{H}_1,{\rm i}\widehat{H}_2,{\rm i}\widehat{H}_3$.  
	
	{As} $SL(2,\mathbb{R})$ has a fundamental group isomorphic to $\mathbb{Z}$, its universal covering, $\widetilde{SL}(2,\mathbb{R})$, is not a matrix Lie group, which is complicated to deal with in practical applications (cf. \cite[p. 127]{Ha15}).  
	
	In applications, one is frequently interested in solutions of (\ref{equation})  for the variable $t$ taking values close to zero. By the Ado theorem, every finite-dimensional Lie algebra can be written in a matrix form and, then, it admits a matrix Lie group \cite{Ha15}. Since all Lie groups with the same Lie algebra are {locally} diffeomorphic close to the neutral element, one can define an analogue of (\ref{equation}) on the matrix Lie group $\mathfrak{sl}(2,\mathbb{R})$ given by 	\begin{equation*}
		\frac{dg}{dt}=-e^{-2\gamma t}X_2^R(g)-\omega^2e^{2\gamma t}X_1^R(g),\qquad g\in {SL}(2,\mathbb{R}),\qquad t\in\mathbb{R}.
	\end{equation*}

	\section{Quantum quasi-Lie schemes}\label{Sec:QQLS}
	
	This section develops the theory of quantum quasi-Lie schemes as an extension of the theory of quasi-Lie schemes \cite{CGL08, Dissertationes}. Its main aim is to investigate the integrability properties of a class of $t$-dependent Schr\"odinger equations containing the $t$-dependent Schr\"odinger equations studied via quantum Lie systems as a particular case. 

	The following example illustrates the necessity of quantum quasi-Lie schemes.
	Consider an $n$-dimensional nonlinear quantum oscillator with a $t$-dependent anharmonic potential described by 
	\begin{equation}\label{non}
		\widehat{H}_{NH}(t):=\frac 12\sum_{i=1}^n \widehat {p}_i^2+\omega(t)\frac12\sum_{i=1}^n\widehat{x}^2_i+\tilde{\omega}(t)\sum_{i=1}^n\widehat{x}_i^\alpha,
	\end{equation}where {$\alpha$} is a negative integer, the term $NH$ stands for 'non-harmonic', and  $\omega(t),\tilde{\omega}(t)$ are real $t$-dependent functions such that the points of the curve $(\omega(t),\tilde{\omega}(t))$ in $\mathbb{R}^2$ span the whole space $\mathbb{R}^2$. 
	Let us prove that $\widehat{H}_{NH}(t)$ is not a quantum Lie system through the following no-go proposition.
	
	\begin{proposition}{\bf (Quantum Lie systems no-go proposition)} Let $\widehat{H}(t)$ be a $t$-dependent Hermitian operator on $\mathcal{H}$ such that there exist $\widehat{H}_a,\widehat{H}_b\in \langle \widehat H(t)\rangle_{t\in\mathbb{R}}$ and a Hermitian operator $\widehat{H}_c$ on $\mathcal{H}$ satisfying  
		$$
		[{\rm i}\widehat{H}_c,{\rm i}\widehat{H}_a]=c_a{\rm i}\widehat{H}_a,\qquad [{\rm i}\widehat{H}_c,{\rm i}\widehat{H}_b]=c_b{\rm i}\widehat{H}_b ,$$ 
		for constants $c_a\in \mathbb{R}$ and $c_b\in \mathbb{R}\backslash \{0\}$. Let ${\rm pr}:V\subset {\rm End}(\mathcal{H})\rightarrow {\rm End}(\mathcal{H})$ be a 
		Lie algebra morphism from a Lie subalgebra $V$ of the complex Lie algebra ${\rm End}(\mathcal{H})$ of operators on $\mathcal{H}$ containing $\langle {\rm i}\widehat H(t)\rangle_{t\in\mathbb{R}}$ and ${\rm i}\widehat{H}_c$.   
		If  ${\rm ad}^n_{{\rm pr} ({\rm i}\widehat{H}_b)}\pi ({\rm i}\widehat{H}_a)\neq 0$ for every natural number $n> 0$, then $\widehat{H}(t)$ is not a quantum Lie system.   
	\end{proposition}
	\begin{proof} By induction, we obtain that $[{\rm i}\widehat{H}_c,{\rm ad}^n_{{\rm i}\widehat{H}_b}{\rm i}\widehat{H}_a]=(nc_b+c_a){\rm ad}^n_{{\rm i}\widehat{H}_b}{\rm i}\widehat{H}_a$ for every $n\in \mathbb{N}$.  Since  ${\rm ad}^n_{{\rm pr} ({\rm i}\widehat{H}_b)}\pi ({\rm i}\widehat{H}_a)\neq 0$ by assumption,  $\pi$ is a Lie algebra morphism, and all the ${\rm i}\widehat{H}_{n}:={\rm ad}^n_{{\rm i}\widehat{H}_b}{\rm i}\widehat{H}_a$ with $n\in \mathbb{N}\cup\{0\}$ belong to the domain of $\pi$, then $\widehat H_{n}\neq 0$ for all $n\in \mathbb{N}$  and the operators $\{{\rm i}\widehat{H}_n\}_{n\in \mathbb{N}}$ are eigenvectors with different eigenvalues of the linear morphism ${\rm ad}_{{\rm i}\widehat{H}_c}:E\rightarrow E$ where $E:=\langle {\rm i}\widehat{H}_n\rangle_{n\in\mathbb{N}}$ is an $\mathbb{R}$-linear space. Hence, the $\{{\rm i}\widehat H_{n}\}_{n\in \mathbb{N}}$ are linearly independent and $\dim \,E=+\infty$. In consequence, every Lie algebra of operators containing the ${\rm i}\widehat{H}(t)$ must be infinite-dimensional since it contains ${\rm i}\widehat{H}_a,{\rm i}\widehat{H}_b$ and all the $\{{\rm i}\widehat {H}_n\}_{n\in \mathbb{N}}$.
	\end{proof}
	
	Let us apply the results of the above proposition to  (\ref{non}). Define the Hermitian operators on $\mathcal{L}^2(\mathbb{R}^n)$ of the form 
	$$
	\widehat{H}_a:=\sum_{j=1}^n\widehat x_j^\alpha,\quad \widehat{H}_b:=\sum_{j=1}^n\widehat {p}_j^2,\quad \widehat{H}_c:=\sum_{j=1}^n\frac{(\widehat x_j\widehat p_j+\widehat p_j\widehat x_j)}{2}.
	$$
	
	Let us define 
	$$
	{\rm pr}({\rm i}\widehat{H}_{NH}(t))=\frac {\rm i}2\widehat{p}_1^2+\omega(t)\frac {\rm i}2 \widehat{x}_1^2+{\rm i}\tilde{\omega}(t)\widehat{x}_1^\alpha,\,  {\rm pr}({\rm i}\widehat{H}_c)={\rm i}\frac{(\widehat x_1\widehat p_1+\widehat p_1\widehat x_1)}{2},\,   {\rm pr}({\rm i}\widehat{H}_a)={\rm i}\widehat{x}_1^\alpha,\,    {\rm pr}({\rm i}\widehat{H}_b)={\rm i}\widehat{p}_1^2.
	$$ 
	Since $[{\rm i}\widehat{H}_c,{\rm i}\widehat{H}_a]=\alpha{\rm i}\widehat{H}_a$, $[{\rm i}\widehat{H}_c,{\rm i}\widehat{H}_b]=-2{\rm i}\widehat{H}_b$ and 	$$
	\left[{\rm i}\widehat p_1^2,{\rm i}\widehat x^{\alpha}_1\widehat p_1^k\right]=2\alpha{\rm i}\widehat x^{\alpha-1}_1\widehat p_1^{k+1}+\alpha(\alpha-1)\widehat {x}_1^{\alpha-2}\widehat {p}_1^{k},
	$$
	it follows that $$
	[{\rm pr}({\rm i}\widehat H_b),{\rm pr}({\rm i}\widehat H_a)]=2\alpha{\rm i}\widehat x^{\alpha-1}_1\widehat p_1+\alpha(\alpha-1)\widehat x_1^{\alpha-2},
	$$
	and, by induction and recalling that $\alpha<0$, we obtain
	$$
	{\rm ad}_{{\rm pr}({\rm i}\widehat H_b)}^n(\pi({\rm i}\widehat H_a))=\sum_{k=0}^na_k\widehat x^{\alpha-(n+k)}_1\widehat p_1^{n-k}, 
	$$
	for certain complex constants $a_0,\ldots,a_n$. In particular, it can be proved that  $a_0=2^ni\alpha(\alpha-1)\cdots (\alpha-n+1)$, while $a_n=(-i)^{n-1}\alpha(\alpha-1)\cdots(\alpha-2n+1)$,
	which are different from zero since $\alpha<0$.  Since ${\rm ad}_{{\rm pr}({\rm i}\widehat H_b)}^n{\rm pr}({\rm i}\widehat H_a)\neq 0$, the quantum Lie systems no-go proposition ensures that (\ref{non}) is not a quantum Lie system.
	
	Consider now a  $t$-dependent Schr\"odinger equation
	\begin{equation}\label{QLS}
		\frac{\partial \psi}{\partial t}=-{\rm i}\widehat{H}(t)\psi,\qquad \psi\in \mathcal{H},\qquad t\in \mathbb{R},
	\end{equation}
	where $-{\rm i}\widehat{H}(t)$ is a $t$-dependent operator taking values in a finite-dimensional real linear space $\mathfrak{V}$ of skew-Hermitian operators. 
	Hence, $-{\rm i}\widehat{H}(t)$ can be understood as a curve in $\mathfrak{V}$.
	
	This time we  do not impose the skew-Hermitian operators $\{-{\rm i}\widehat{H}(t)\}_{t\in\mathbb{R}}$ and their successive commutators to span a real finite-dimensional Lie algebra of operators (with respect to the operator commutator as in the case of quantum Lie systems). Instead,  suppose that $-{\rm i}\widehat{H}(t)$ takes values in a finite-dimensional  {real linear}  space of skew-Hermitian operators $\mathfrak{V}$ and we also assume that there exists a non-zero real Lie algebra $\mathfrak{W}\subset \mathfrak{V}$ such that  $[\mathfrak{W},\mathfrak{V}]\subset \mathfrak{V}$. 
	
	\begin{definition} A {\it quantum quasi-Lie scheme}, $S(\mathfrak{W},\mathfrak{V})$, is a pair of non-zero  finite-dimensional $\mathbb{R}$-linear spaces $\mathfrak{W},\mathfrak{V}$ of skew-Hermitian operators on a Hilbert space $\mathcal{H}$ satisfying that
		\begin{equation}\label{con}
			\mathfrak{W}\subset \mathfrak{V},\qquad  [\mathfrak{W},\mathfrak{W}]\subset \mathfrak{W}, \qquad [\mathfrak{W},\mathfrak{V}]\subset \mathfrak{V}.
		\end{equation}
	\end{definition}
	
	
	To illustrate quasi-Lie schemes, let us consider the $\mathbb{R}$-linear space of skew-Hermitian operators 
	$$
	\mathfrak{V}_{NH}:=\left\langle {\rm i}\widehat{H}_0:={\rm i}\sum_{j=1}^n\widehat {p}_j^2,\,\,{\rm i}\widehat{H}_1:={\rm i}\sum_{j=1}^n\widehat x_j^\alpha,\,\,{\rm i}\widehat{H}_2:={\rm i}\sum_{j=1}^n\widehat x^2_j,\,\,   {\rm i}\widehat{H}_3:={\rm i}\sum_{j=1}^n \frac{(\widehat x_j\widehat p_j+\widehat p_j\widehat x_j)}{2}\right\rangle,
	$$
	for a fixed $\alpha\in \mathbb{N}\backslash\{2\}$ and $\mathfrak{W}_{NH}:=\langle {\rm i}\widehat{H}_2,{\rm i}\widehat{H}_3\rangle$. 
	Since $[{\rm i}\widehat{H}_2,{\rm i}\widehat{H}_3]=-2{\rm i}\widehat{H}_2$, the linear space $\mathfrak{W}_{NH}$ is a {two-dimensional real} Lie algebra. {We can see that} $[\mathfrak{W}_{NH},\mathfrak{V}_{NH}]\subset \mathfrak{V}_{NH}$ since
	$$
	[{\rm i}\widehat{H}_2,{\rm i}\widehat{H}_1]=0,\qquad [{\rm i}\widehat{H}_2,{\rm i}\widehat{H}_0]=-{\rm i}\widehat{H}_3,\qquad [{\rm i}\widehat{H}_3,{\rm i}\widehat{H}_0]=-2{\rm i}\widehat{H}_0,\qquad [{\rm i}\widehat{H}_3,{\rm i}\widehat{H}_1]=\alpha{\rm i}\widehat{H}_1.
	$$

	It follows that $\mathfrak{W}_{NH}$ and $\mathfrak{V}_{NH}$ satisfy the conditions (\ref{con}) and  {therefore} the pair $\mathfrak{W}_{NH},\mathfrak{V}_{NH}$  gives rise to a quantum quasi-Lie scheme $S(\mathfrak{W}_{NH},\mathfrak{V}_{NH})$. Moreover, $-{\rm i}\widehat{H}_{NH}(t)$ takes values in $\mathfrak{V}_{NH}$. Indeed,
	$$
	\widehat H_{NH}(t)=\frac 12\widehat{H}_0+\omega(t)\frac12\widehat{H}_2+\tilde{\omega}(t)\widehat{H}_1.
	$$

	The last expression will be a clue to study the $t$-dependent Schr\"odinger equation related to $\widehat{H}_{NH}(t)$ through quasi-Lie schemes. Indeed, it suggests us, along with other examples to be studied in the forthcoming sections, to propose the following definition.
	
	%
	
	\begin{definition}\label{Def:CQLS}
		The $t$-dependent Schr\"odinger equation (\ref{QLS}) admits a {\it compatible quantum quasi-Lie scheme} $S(\mathfrak{W},\mathfrak{V})$ if 
		$-{\rm i}\widehat{H}(t)$ is a curve taking values in $\mathfrak{V}$.
	\end{definition}
	Let us stress that the idea behind Definition \ref{Def:CQLS} is that given a $t$-dependent Schr\"odinger equation (\ref{QLS}), we can apply the theory of quasi-Lie schemes to study it only via a quasi quasi-Lie scheme that is compatible with it in the sense given above. If not otherwise stated, we hereafter assume that every quasi-Lie scheme and $t$-dependent skew-Hermitian operator is defined on a generic (possibly infinite-dimensional) Hilbert space $\mathcal{H}$.
	
	\section{Quantum quasi-Lie systems}\label{Sec:QQLSy}
	
	Many {of the} works in the literature try to map a certain $t$-dependent operator into a simpler one by means of a gauge transformation (see e.g. \cite{Su97}). Similarly, we introduce quantum quasi-Lie systems in this section as a class of  $t$-dependent Schr\"odinger equations with a compatible quantum quasi-Lie scheme that can be transformed into a $t$-dependent Schr\"odinger equation described via a quantum-Lie system. To achieve this goal, it becomes necessary to extend  some previous results on quantum Lie systems.

	\begin{definition}The {\it representation of a quasi-Lie scheme} $S(\mathfrak{W},\mathfrak{V})$ is a Lie algebra morphism 
		$$
		\begin{array}{rccc}
			\rho&:\mathfrak{W}&\rightarrow &{\rm End}(\mathfrak{V}),\\
			&\widehat{A}&\mapsto &\rho_{\widehat{A}}:=[\widehat{A},\cdot],
		\end{array}
		$$
		where we recall that ${\rm End}(\mathfrak{V})$ is the Lie algebra of endomorphisms on $\mathfrak{V}$ relative to the commutator of endomorphisms. A {\it subrepresentation} of $S(\mathfrak{W},\mathfrak{V})$ is a subspace $\mathfrak{V}_1\subset \mathfrak{V}$ such that $S(\mathfrak{W},\mathfrak{V}_1)$ is a quasi-Lie scheme.
	\end{definition}

	As the subspace $\mathfrak{W}$ of a quasi-Lie scheme $S(\mathfrak{W},\mathfrak{V})$ is a real Lie algebra of skew-Hermitian operators, 
	there exists a Lie group action $\Phi:(g,\psi)\in G \times \mathcal{H}\mapsto \Phi(g,\psi)\in \mathcal{H}$ such that the mappings $\Phi^g:\psi\in \mathcal{H}\mapsto \Phi(g,\psi)\in \mathcal{H} $, with $g\in G$, are unitary operators on $\mathcal{H}$ and, moreover, one has that 
	$$
	\frac{d}{dt}\bigg|_{t=t_0}\Phi(\exp(tv),\psi)=\widetilde{\rho}(v)\Phi(\exp(t_0v),\psi),\qquad \forall \psi\in \mathcal{H},
	$$
	where $\mathfrak{g}$ is the Lie algebra of $G$, while $\widetilde{\rho}:\mathfrak{g}\rightarrow \mathfrak{W}$ is a Lie algebra isomorphism relative to the commutator operator on elements of $\mathfrak{W}$. %
	Note  {that there} may be other Lie group action $\Phi':G'\times \mathcal{H}_{\mathbb{R}}\rightarrow \mathcal{H}_{\mathbb{R}}$, where the Lie algebra $\mathfrak{g}'$ of $G'$ is isomorphic to $\mathfrak{g}$ and $G'$ is connected, satisfying similar properties. Let us study the relations between such Lie group actions.
	
	Let as consider $\widehat{\Phi}:g\in G\mapsto \Phi^g\in U(\mathcal{H})$ and $\widehat{\Phi}':g\in G'\mapsto \Phi'^g\in U(\mathcal{H})$. Now, we aim to show that $\widehat{\Phi}(G)=\widehat{\Phi}'(G')$.  Due to the fact that $\Phi$ and $\Phi'$ are Lie group actions and $\widetilde{\rho}$, $\widetilde{\rho}\,'$ are their tangent maps at $0$, respectively, one can consider the following commutative diagram
	\begin{equation}\label{Dia}
		\xymatrix{\mathfrak{g}\ar[r]^{\widetilde{\rho}}\ar[d]^{\exp_G}&\mathfrak{W}\ar[d]^\exp&\mathfrak{g}'\ar[l]_{\widetilde{\rho}\,'}\ar[d]^{\exp_{G'}}\\G\ar[r]^{\widehat{\Phi}}&U(\mathcal{H})&G'\ar[l]_{\widehat{\Phi}'}}
	\end{equation}
	where $\exp_G$, $\exp_{G'}$ are the exponential maps from $\mathfrak{g},\mathfrak{g}'$ into the Lie groups $G,G'$, respectively, and $\exp$ stands for the exponential of operators in $\mathfrak{W}$, which exists because the elements of $\mathfrak{W}$ are skew-Hermitian operators.  Since $G$ is connected, every element of $G$ can be written as a product of elements in the image of the $\exp_G$ \cite[p. 224]{Sc94}. Hence, the group generated by the elements of $\widehat{\Phi}(\exp_G(\mathfrak{g}))$ is $\widehat{\Phi}(G)$. Repeating the same argument concerning the right-hand side of diagram (\ref{Dia}), using that it is commutative, and the fact that ${\rm Im}\,\widetilde{\rho}={\rm Im}\,\widetilde{\rho}\,'$, we obtain that $\widehat{\Phi}(G)=\widehat{\Phi}'(G')$.
	Note that $\widehat{\Phi}(G)$ is not a Lie subgroup of $U(\mathcal{H})$ even when $\mathcal{H}$ is finite-dimensional and $U(\mathcal{H})$ becomes a standard Lie group (cf. \cite{CMP17}).  Moreover, the space $\mathcal{G}_\mathfrak{W}$ of curves $U(t)$ in $\Phi(G)$ with $U(0)={\rm Id}_\mathcal{H}$ is a group  relative to the multiplication 
	$$
	(U_1\star U_2)(t)=U_1(t)U_2(t),\qquad U_1(t),U_2(t)\in\mathcal{G}_\mathfrak{W},\qquad \forall t\in \mathbb{R}. 
	$$
	These ideas justify the following definition.

	\begin{definition} If $S(\mathfrak{W},\mathfrak{V})$ is a quasi-Lie scheme and $\Phi:G_\mathfrak{W}\times \mathcal{H}\rightarrow \mathcal{H}$ is a Lie group action obtained by integrating the Lie algebra $\mathfrak{W}$, we call {\it group} of the quasi-Lie scheme $S(\mathfrak{W},\mathfrak{V})$ the space $\mathcal{G}_\mathfrak{W}$ of curves $U(t)$ in $\Phi(G)$ with $U(0)={\rm Id}_\mathcal{H}$.
	\end{definition}
	
	Recall that, given a $t$-dependent Schr\"odinger equation (\ref{QLS}), one can define a family of unitary operators $U(t_2,t_1)$, with $t_2,t_1\in \mathbb{R}$, such that, given a certain $\psi_0\in \mathcal{H}$, then $U(t_2,t_1)\psi_0$ is the value of the particular solution $\psi(t)$ to (\ref{QLS}) for $t=t_2$ with initial condition $\psi(t_1)=\psi_0$.
	To simplify the notation, we will call $t$-dependent evolution operator of a Schr\"odinger equation (\ref{QLS}) the $t$-dependent family of unitary operators, $U(t)$,  such that $U(t)=U(t,0)$.

	\begin{proposition}\label{CharQLS} 	The elements of the {\it group} $\mathcal{G}_\mathfrak{W}$ of a quasi-Lie scheme $S(\mathfrak{W},\mathfrak{V})$ are the $t$-dependent evolution operators  of the $t$-dependent Schr\"odinger equations of the form
		\begin{equation}\label{GS}
			\frac{\partial \psi}{\partial t}=-{\rm i}\widehat H_0(t)\psi,\qquad \psi\in \mathcal{H},
		\end{equation}
		for a certain $t$-dependent operator $-{\rm i}\widehat H_0(t)$ taking values in $\mathfrak{W}$. Conversely, the $t$-dependent evolution operator to (\ref{GS}) can be described by a $t$-dependent evolution operator $U(t)$ in $\mathcal{G}_\mathfrak{W}$. 
	\end{proposition}
	\begin{proof}Let $\mathcal{U}$ be an open neighbourhood of the neutral element of the Lie group $G_\mathfrak{W}$, where canonical coordinates of the second-kind can be defined. 
		Consider a curve $U(t)$ in $\mathcal{G}_\mathfrak{W}$. Then, $U(0)={\rm Id}_\mathcal{H}$ and, if  $t$ is so close enough to $0$, let us say $t\in(-\epsilon,\epsilon)$ for a certain real $\epsilon>0$, we can assume that $U(t)$ takes values in $\widehat{\Phi}(\mathcal{U})$ for $t\in(-\epsilon,\epsilon)$, where we recall that $\widehat{\Phi}$ is the morphism of groups
		$\widehat{\Phi}:G_{\mathfrak{W}}\rightarrow U(\mathcal{H})$ related to the Lie group action $\Phi:G_\mathfrak{W}\times \mathcal{H}\rightarrow \mathcal{H}$ induced by the integration of $\mathfrak{W}$. Using canonical coordinates of second-kind in $\mathcal{U}$ in a particular basis $\{{\rm i}\widehat{H}_1,\ldots,{\rm i}\widehat{H}_r\}$ of $\mathfrak{W}$, we obtain that $U(t)$ can be written as  
		\begin{equation}\label{curve}
			U(t)=\exp\left({\rm i}f_1(t) \widehat{H}_1\right)\times\ldots\times\exp\left({\rm i}f_r(t) \widehat{H}_r\right){,}
		\end{equation}
		for certain $t$-dependent real  functions $f_1(t),\ldots, f_r(t)$ vanishing at $t=0$.

		A short calculation shows that 
		\begin{equation}\label{eq6}
			\frac{dU}{dt}(t)\psi= \sum_{j=1}^r\frac{df_j}{dt}(t){\rm Ad}_{ \prod_{k=1}^{j-1}\exp({\rm i}f_k(t)\widehat H_k) }({\rm i}\widehat H_j)U(t)\psi,\qquad \forall \psi\in \mathcal{H},
		\end{equation}
		where ${\rm Ad}_{\exp({\rm i}f_k(t)\widehat H_k)}\widehat H_j:=\exp({\rm i}f_k(t)\widehat H_k)\widehat H_j\exp(-{\rm i}f_k(t)\widehat H_k)$ for every $j,k=1,\ldots,r$ and $t\in \mathbb{R}$. Note that the exponentials are arranged in ascending order relative to $k$. We also assume that ${\rm Ad}_{\prod_{k=1}^{0}\exp({\rm i}f_k(t)\widehat H_k)}:={\rm Id}_\mathfrak{W}$. As $\mathfrak{W}$ is a {real} Lie algebra, each morphism ${\rm Ad}_{\exp(\lambda {\rm i}\widehat H_j)}$, with $\lambda\in\mathbb{R}$ and $j=1,\ldots,r$, leaves $\mathfrak{W}$ invariant. Moreover, the right-hand side of (\ref{eq6}) shows that 
		$$
		\frac{dU}{dt}(t)U(t)^{-1}=\sum_{j=1}^r\frac{df_j}{dt}(t){\rm Ad}_{ \prod_{k=1}^{j-1}\exp({\rm i}f_k(t)\widehat H_k) }({\rm i}\widehat H_j)  
		$$ becomes a curve $-{\rm i}\widehat{H}_0(t)$ taking values in  $\mathfrak{W}$. Consequently, $U(t)$ determines a curve  $-{\rm i}\widehat{H}_0(t)$ in $ \mathfrak{W}$ and the general solution to its $t$-dependent Schr\"odinger equation  has general solution $U(t)\psi_0$ for $\psi_0\in \mathcal{H}$. Recall that since $f_k(0)=0$, for $k=1,\ldots,r$, the $\psi_0\in \mathcal{H}$ is the initial condition for the particular solution $\psi(t)=U(t)\psi_0$ and $U(0)={\rm Id}_\mathcal{H}$.
		
		Conversely, assume that we are given a curve $-{\rm i}\widehat H_0(t)$  taking values in $\mathfrak{W}$. Let us determine a curve $U(t)$, defined at least for $t$ in a neighbourhood around zero, describing the $t$-dependent evolution operator for (\ref{GS}). Around $t=0$, the curve $U(t)$ can be written in the form (\ref{curve}) because $U(0)={\rm Id}_{\mathcal{H}}$. 
		We can define a $t$-dependent family of   mappings of the form $T_t:\mathfrak{W}\rightarrow \mathfrak{W}$, for each $t\in \mathbb{R}$, such that
		$$
		-{\rm i}\sum_{j=1}^r \lambda_j\widehat {H}_j\stackrel{T_t}{\mapsto} -\sum_{j=1}^r\lambda_j{\rm Ad}_{ \prod_{k=1}^{j-1}\exp({\rm i}f_k(t)\widehat H_k ) }({\rm i}\widehat H_j)
		$$
		for every set $\lambda_1,\ldots,\lambda_r\in \mathbb{R}$. 
		The functions $f_1(t),
		\ldots, f_r(t)$ are undetermined now, but $f_1(0)=\ldots=f_r(0)=0$ and then  $T(0)={\rm Id}_{\mathfrak{W}}$. Since $T_t$ depends continuously on $t$, then each particular $T_{t_0}$ will be invertible for $t_0$ in a  {neighbourhood}  of zero. As a consequence, the system of differential equations given by
		$$
		-\sum_{j=1}^r\frac{df_j}{dt}(t){\rm Ad}_{\left(\prod_{k=1}^{j-1}\exp({\rm i}f_k(t)\widehat H_k\right)}({\rm i}\widehat H_j)=-{\rm i}\widehat{H}_0(t)
		$$
		amounts, after applying $T^{-1}_t$, to
		$$
		-{\rm i}\sum_{k=1}^r\frac{df_k}{dt}(t)\widehat {H}_k=T^{-1}_t(-{\rm i}\widehat H_0(t)),
		$$
		which always has a local solution $f_1(t),\ldots,f_r(t)$ for $t$ in a  {neighbourhood} of zero. This allows us to determine the values of $f_1(t),\ldots,f_r(t)$ close to $t=0$ ensuring that $U(t)$ is the $t$-dependent evolution operator of (\ref{GS}). 
		
	\end{proof}

	The main property of a quantum quasi-Lie scheme  is given in the following theorem, which will be employed to simplify $t$-dependent Schr\"odinger equations 'compatible' with quantum quasi-Lie schemes. 
	
	\begin{definition} Let $S(\mathfrak{W},\mathfrak{V})$ be a quasi-Lie scheme and let $U(t)$ be the $t$-dependent evolution operator related to a $t$-dependent skew-Hermitian operator $-{\rm i}\widehat{H}(t)$ taking values in $\mathfrak{V}$. Then, we write $U'(t)_{\di}(-{\rm i}\widehat{H}(t))$, where $U'(t)$ stands for a curve in $\mathcal{G}_{\mathfrak{W}}$, for the $t$-dependent skew-Hermitian operator $-{\rm i}\widehat{H}'(t)$ associated with the evolution operator $(U'\star U)(t)$. We also write $\mathfrak{W}_{\mathbb{R}}$ and $\mathfrak{V}_{\mathbb{R}}$ for the spaces of curves in $\mathfrak{W}$ and $\mathfrak{V}$, respectively.
		
	\end{definition}
	
	\begin{theorem}\label{MainT} {\bf (The main theorem of quantum Lie systems)} Let $S(\mathfrak{W},\mathfrak{V})$ be a quantum quasi-Lie scheme  and let  $-{\rm i}\widehat{H}(t)$ be a curve in  $\mathfrak{V}$. 
		If $U_\mathfrak{W}(t)$ is an element of $
		\mathcal{G}_\mathfrak{W}$,  then
		\begin{equation}\label{action}
			U_\mathfrak{W}(t)_\di(-{\rm i}\widehat{H}(t)):=\widehat A(t)+{\rm Ad}_{U_\mathfrak{W}(t)}(-{\rm i}\widehat H(t)),
		\end{equation}
		where $\widehat {A}(t)$ is the $t$-dependent skew-Hermitian operator  taking values in $\mathfrak{W}$ whose evolution is determined by  $U_\mathfrak{W}(t)$, is an element of $\mathfrak{V}_\mathbb{R}$. If $\psi(t)$ is a particular solution to the $t$-dependent Schr\"odinger equation related to $-{\rm i}\widehat H(t)$, then the curve in $\mathcal{H}$ of the form $\psi'(t)=U_\mathfrak{W}(t)\psi(t)$ is a solution to the Schr\"odinger equation related to $U_\mathfrak{W}(t)_\di(-{\rm i}\widehat{H}(t))$.
	\end{theorem}
	\begin{proof} As $U_\mathfrak{W}(t)$ is the $t$-dependent evolution operator related to the $t$-dependent skew-Hermitian operator $\widehat A(t)$, one has that $\widehat A(t)$ is an element of $ \mathfrak{W}_\mathbb{R}\subset \mathfrak{V}_\mathbb{R}$. 
		Since $-{\rm i}\widehat{H}'(t):=U_\mathfrak{W}(t)_\di(-{\rm i}\widehat{H}(t))$ is by definition the $t$-dependent skew-Hermitian operator admitting a $t$-dependent evolution operator $U'(t):=U_\mathfrak{W}(t)U(t)$, where $U(t)$ is the $t$-dependent evolution operator for $\widehat{H}(t)$, then
		\begin{multline}\label{Gau}
			\frac{dU'}{dt}(t)U'^\dagger(t)\psi_0=\left[\frac{dU_{\mathfrak{W}}}{dt}(t)U_{\mathfrak{W}}^\dagger(t)+{\rm Ad}_{U_{\mathfrak{W}}(t)}
			\left(\frac{dU}{dt}(t)U^\dagger(t)\right)\right]\psi_0\\=\widehat{A}(t)\psi_0-{\rm Ad}_{U_{\mathfrak{W}}(t)}
			\left({\rm i}\widehat{H}(t)\right)\psi_0
			=\widehat{A}(t)\psi_0+\left(\sum_{n=0}^\infty \frac{1}{n!}{\rm ad}^n_{\widehat{A}(t)}(-{\rm i}\widehat{H}(t))\right)\psi_0,
		\end{multline}
		for every $\psi_0\in \mathcal{H}$. 
		As $S(\mathfrak{V},\mathfrak{W})$ is a quantum Lie system, the curve $-[\widehat{B},
		{\rm i}\widehat{H}(t)]\subset \mathfrak{V}$, for any $\widehat{B}\in \mathfrak{W}$, is a curve in $\mathfrak{V}$ and $-{\rm i}\widehat{A}(t)$ takes values in $\mathfrak{W}$. By induction, ${\rm ad}^n_{\widehat{A}(t)}({\rm i}\widehat{H}(t))$ takes values in $\mathfrak{V}$ for every $n\in \mathbb{N}\cup \{0\}$ and, in view of (\ref{Gau}), one has that 
		$$\frac{dU'}{dt}(t)U'^{\dagger}(t)=-{\rm i}\widehat{H}'(t)$$
		takes values in $ \mathfrak{V}.$
	\end{proof}
	
	Summarising, a quantum quasi-Lie scheme induces a group of $t$-dependent transformations allowing us to transform the $t$-dependent Schr\"odinger
	equation described by a $t$-dependent skew-Hermitian operator taking values in  the linear space $\mathfrak{V}$ of a quantum quasi-Lie scheme $S(\mathfrak{W},\mathfrak{V})$ into 
	another $t$-dependent Schr\"odinger equations related to $t$-dependent skew-Hermitian operators taking values also in $\mathfrak{V}$. A particular relevant case occurs when the quantum quasi-Lie scheme enables us to transform, via an element of $\mathcal{G}_{\mathfrak{W}}$, the initial $t$-dependent Schr\"{o}dinger equation into a final one which can be described by means of the usual
	theory of quantum Lie systems.
	This motivates the next definition.
	
	\begin{definition} The $t$-dependent skew-Hermitian operator $-{\rm i}\widehat{H}(t)$ is a {\it quantum quasi-Lie system} with respect to the quantum quasi-Lie scheme
		$S(\mathfrak{W},\mathfrak{V})$ if $-{\rm i}\widehat{H}(t)$ takes values in $\mathfrak{V}$ and there exists a curve $U(t)$ in $ \mathcal{G}_\mathfrak{W}$ such that $
		U(t)_{\di}(-{\rm i}\widehat{H}(t))$ is a quantum Lie system.
	\end{definition}
	
	To simplify the notation, we will call quantum quasi-Lie system and quantum Lie system the $t$-dependent Schr\"odinger equations related to a quantum quasi-Lie system or a quasi-Lie system, respectively. 
	
	\section{Transformation properties of quantum quasi-Lie schemes}\label{Sec:TPQQLS}
	The group, $\mathcal{G}_{\mathfrak{W}}$, of the quantum quasi-Lie scheme $S(\mathfrak{W},\mathfrak{V})$ allows us to map every $t$-dependent skew-Herminitan operator $-{\rm i}\widehat{H}(t)$ taking values in $\mathfrak{V}$ into new ones taking values in $\mathfrak{V}$. This fact {may be}  helpful in integrating or, at least, simplifying $t$-dependent Schr\"odinger equations  related to $t$-dependent Hamiltonian operators \cite{Su97}. This section depicts new easily implementable techniques to accomplish such a study.

	\begin{proposition}\label{Red1} If $S(\mathfrak{W},\mathfrak{V})$ is a quantum quasi-Lie scheme, $-{\rm i}\widehat{H}(t)$ takes values in $\mathfrak{V}$,  and  $P:\mathfrak{V}\rightarrow \mathfrak{W}$ is any projection of $\mathfrak{V}$ onto $\mathfrak{W}$, then there exists $U(t)$ in $ \mathcal{G}_\mathfrak{W}$ such that  $P[U(t)_\di (-{\rm i}\widehat{H}(t))]=0$ at least for $t$ in an open  {neighbourhood} of $0$ in $\mathbb{R}$.
	\end{proposition}
	\begin{proof} An element $U(t)$ in $\mathcal{G}_\mathfrak{W}$ acts on $-{\rm i}\widehat H(t)$ taking values in $ \mathfrak{V}$ in the form (\ref{action}). Applying the projection $P$ in both sides of (\ref{action}) and using that $d U(t)/dt(t)U^{\dagger}(t)$ takes values in $\mathfrak{W}$ for every $t\in \mathbb{R}$, one obtains
		$$
		P[U(t)_\di (-{\rm i}\widehat H(t))]=P\left(\frac{dU}{dt}(t) U^{\dagger}(t)-{\rm Ad}_{U(t)}({\rm i}\widehat H(t))\right)=\frac{dU}{dt}(t) U^{\dagger}(t)-P({\rm Ad}_{U(t)}({\rm i}\widehat H(t))).
		$$
		Consider the differential equation {for $U(t)$} of the form
		\begin{equation}\label{Prob}
			\frac{dU}{dt}(t)U^\dagger(t)=P({\rm Ad}_{U(t)}{\rm i}\widehat H(t)).
		\end{equation}
		Any particular solution $U(t)$ of this equation with $U(0)={\rm Id}_\mathcal{H}$ will map $-{\rm i}\widehat H(t)$ into a new $t$-dependent skew-Hermitian operator $-{\rm i}\widehat H'(t):=U(t)_\di(-{\rm i}\widehat{H}(t))$ such that $P[U(t)_\di (-{\rm i}\widehat H(t)]=0$. 
		
		Let us consider a linear space, $\mathcal{V}$, isomorphic to $\mathfrak{V}$ via $\Theta:\mathcal{V}\rightarrow \mathfrak{V}$. Consider also $\mathcal{W}\subset \mathcal{V}$ to be the linear subspace $\mathcal{V}:=\Theta^{-1}(\mathfrak{W})$. Then,  $\mathcal{W}$ inherits via $\Theta$ a Lie algebra structure that makes that $\mathcal{W}$ can be considered as the Lie algebra of a Lie group  $G_\mathfrak{W}$.  
		The isomorphism $\Theta:\mathcal{V}\rightarrow \mathfrak{V}$  allows us to define a projector $\mathcal{P}:\mathcal{V}\rightarrow \mathcal{V}$ onto $\mathcal{W}$ of the form $\mathcal{P}:=\Theta^{-1}\circ P\circ \Theta$. 
		Consider 
		the system of differential equations on $G_\mathfrak{W}$ of the form
		\begin{equation}\label{Eq:Rel}
			R_{g^{-1}*g}\frac{dg}{dt}=\mathcal{P}({\rm Ad}_ga(t)),
		\end{equation}
		where $a(t):=\Theta^{-1}({\rm i}\widehat{H}(t))$. Let us prove that the solution to this system {satisfying that $g(0)=e$ }  is such that the related $U(t)$ is a solution to (\ref{Prob}). Let $\Phi:G_{\mathfrak{W}}\times\mathcal{H}\rightarrow \mathcal{H}$ be the  Lie group action induced by the Lie algebra isomorphism  $\Theta|_{\mathcal{W}}:\mathcal{W}\rightarrow\mathfrak{W}$. If we define $U(t)=\Phi^{g(t)}$ and $\psi\in \mathcal{H}$, we have that
		$$
		\frac{d}{dt}\bigg|_{t=t_0}\!\!\!U(t)\psi=\frac{d}{dt}\bigg|_{t=t_0}\!\!\!\Phi(g(t)g^{-1}(t_0),\Phi(g(t_0),\psi))=\Theta(\mathcal{P}({\rm Ad}_{g(t_0)}a(t_0)))U(t_0)\psi.
		$$
		Since $\Theta({\rm Ad}_ga(t))={\Phi^g}{\rm i}\widehat{H}(t)\Phi^{g^{-1}}$, we get from the definition of $\mathcal{P}$ that
		$$
		P({\Phi^g}{\rm i}\widehat{H}(t)\Phi^{g^{-1}})=P(\Theta({\rm Ad}_ga(t)))=\Theta(\mathcal{P}({\rm Ad}_ga(t))).
		$$
		Hence,
		$$
		\frac{d}{dt}\bigg|_{t=t_0}U(t)\psi=\Theta(\mathcal{P}({\rm Ad}_ga(t_0)))U(t_0)\psi=P(U(t){\rm i}\widehat{H}(t)U^{-1}(t))U(t)\psi.
		$$
		Since the differential equation (\ref{Eq:Rel}) admits a local solution for every initial condition, {the same is true for the differential equation (\ref{Prob})}.
	\end{proof}
	
	If $\mathfrak{V}$ is a Lie algebra, namely $-{\rm i}\widehat H(t)$ is a quantum Lie system, then we can chose $\mathfrak{W}=\mathfrak{V}$ and the previous proposition ensures that there exists $U(t)$ in $ \mathcal{G}_\mathfrak{W}$ such that $-U(t)_\di{\rm i}\widehat{H}(t)=0$. Hence, $U^{-1}(t)\psi$, for any $\psi\in \mathcal{H}$ becomes the general solution to the $t$-dependent Schr\"odinger equation related to $-{\rm i}\widehat{H}(t)$.	
	
	Let us use  Proposition \ref{Red1}  to simplify or to integrate  $-{\rm i}\widehat{H}(t)$ taking values in $\mathfrak{V}$ when $\mathfrak{V}$ is not a Lie algebra. In such a case, if $-{\rm i}\widehat{H}(t)=\displaystyle{\sum_{\alpha=1}^{r+s}}b_\alpha(t){\rm i}H_\alpha$ for certain $t$-dependent functions $b_1(t),\ldots,b_{r+s}(t)$ and a basis $\{{\rm i}H_1,\ldots,{\rm i}H_{r+s}\}$ of $\mathfrak{V}$ such that the first $r$ elements form a basis for $\mathfrak{W}$, then Proposition \ref{Red1} ensures that there exists an $U(t)$ in $\mathcal{G}_\mathfrak{W}$ such that $U(t)_\di {\rm i}\widehat{H}(t)=\displaystyle{\sum_{\alpha=r+1}^{r+s}}b_\alpha(t){\rm i}H_\alpha$. This simplifies the expression of the initial $t$-dependent Schr\"odinger equation related to $-{\rm i}\widehat{H}(t)$.
	
	To provide further techniques to integrate $t$-dependent Schr\"odinger equations determined by a $t$-dependent skew-Hermitian operator $-{\rm i}\widehat{H}(t)$ in $\mathfrak{V}_\mathbb{R}$ when $\mathfrak{V}$ is not a Lie algebra, let us extend to the quantum case some structures defined for standard quasi-Lie schemes in \cite{CL15}. 

	The result of the following theorem allows us to determine when a $t$-dependent skew-Hermitian  operator, $-{\rm i}\widehat{H}(t)$, taking values in the  space $\mathfrak{V}$ of a quantum quasi-Lie scheme $S(\mathfrak{W},\mathfrak{V})$ cannot be mapped into zero  via an element $U_\mathfrak{W}(t)$ of the group $\mathcal{G}_\mathfrak{W}$ of the quantum quasi-Lie scheme.  In this case, we will use a quantum quasi-Lie scheme  to study the form of $-U(t)_\di{\rm i}\widehat H(t)$ for each $U(t)$ in $ \mathcal{G}_\mathfrak{W}$.
	
	
	\begin{theorem}\label{Red2} Let $S(\mathfrak{W},\mathfrak{V})$ be a quantum quasi-Lie scheme with a  subrepresentation $\mathfrak{V}_1$ of codimension one, namely $\dim \mathfrak{V}/\mathfrak{V}_1=1$, and let $\tau_{\mathfrak{V}_1}:\mathfrak{V}\rightarrow \mathfrak{V}/\mathfrak{V}_1$ be the canonical projection onto $\mathfrak{V}/\mathfrak{V}_1$. 
		If $-{\rm i}\widehat H(t)$ belongs to $\mathfrak{V}_\mathbb{R}$ and it is  such that $\tau_{\mathfrak{V}_1}(-{\rm i}\widehat H(t))\neq 0$ and $\tau_{\mathfrak{V}_1}(\mathfrak{W})=0$, then $\tau_{\mathfrak{V}_1}(-U(t)_\di {\rm i}\widehat H(t))\neq 0$ for every $U(t)$ in $ \mathcal{G}_\mathfrak{W}$.
	\end{theorem}
	\begin{proof} To make the notation clearer, we will simply denote elements of $\mathfrak{W}$ and $\mathfrak{V}$ by lower-case letters, e.g. $w\in \mathfrak{W},v,v_1\in \mathfrak{V}$ and the elements of $\mathcal{G}_\mathfrak{W}$ by curves of the form $g(t)$. Since $[\mathfrak{W},\mathfrak{V}_1]\subset \mathfrak{V}_1$, each Lie algebra morphism $\rho_w=[w,\cdot]$, with $w\in \mathfrak{W}$, induced by the representation of $S(\mathfrak{W},\mathfrak{V})$ leads to  a well-defined morphism $\widehat \rho_w:\tau_{\mathfrak{V}_1}(v)\in \mathfrak{V}/\mathfrak{V}_1\mapsto \tau_{\mathfrak{V}_1}(\rho_w(v))\in \mathfrak{V}/\mathfrak{V}_1$. In fact, if $v_1,v_2\in \mathfrak{V}$ and $\tau_{\mathfrak{V}_1}(v_1)=\tau_{\mathfrak{V}_1}(v_2)$, then $[w,v_1-v_2]\in \mathfrak{V}_1$ for every $w\in \mathfrak{W}$, $\tau_{\mathfrak{V}_1}([w,v_1-v_2])=0$, and
		$$
		\widehat{\rho}_w(\tau_{\mathfrak{V}_1}(v_1)):=\tau_{\mathfrak{V}_1}([w,v_1])=\tau_{\mathfrak{V}_1}([w,v_2])=:\widehat{\rho}_w(\tau_{\mathfrak{V}_1}(v_2)).
		$$
		This allows us to define a Lie algebra morphism $\widehat \rho:w\in \mathfrak{W}\mapsto \widehat{\rho}_w\in  {\rm End}(\mathfrak{V}/\mathfrak{V}_1)$. Indeed,
		$$
		[\widehat \rho_{w_1},\widehat \rho_{w_2}](\tau_{\mathfrak{V}_1}(v)):=\widehat \rho_{w_1}(\widehat \rho_{w_2}(\tau_{\mathfrak{V}_1}(v)))-\widehat \rho_{w_2}(\widehat \rho_{w_1}(\tau_{\mathfrak{V}_1}(v)))=\widehat \rho_{w_1}(\tau_{\mathfrak{V}_1} ([w_2,v]))-\widehat \rho_{w_2}(\tau_{\mathfrak{V}_1} ([w_1,v])),
		$$
		for all $w_1,w_2\in \mathfrak{W}$ and every $v\in \mathfrak{V}$. Using the Jacobi identity and the definition of $\tau_{\mathfrak{V}_1}$,
		$$
		[\widehat \rho_{w_1},\widehat \rho_{w_2}](\tau_{\mathfrak{V}_1}(v)):=\tau_{\mathfrak{V}_1}([w_1,[w_2,v]]-[w_2,[w_1,v]])=\widehat \rho_{[w_1,w_2]}(\tau_{\mathfrak{V}_1}(v)),\qquad \forall w_1,w_2\in \mathfrak{W},\forall v\in \mathfrak{V}.
		$$
		Hence, $[\widehat \rho_{w_1},\widehat \rho_{w_2}]=\widehat \rho_{[w_1,w_2]}$ and $\widehat \rho$ is a Lie algebra morphism. Moreover, $\widehat \rho$ and $\rho:w\in \mathfrak{W}\mapsto [w,\cdot]\in {\rm End}(\mathfrak{V})$ are equivariant relative to $\tau_{\mathfrak{V}_1}$, namely $\widehat\rho_w\circ \tau_{\mathfrak{V}_1}=\tau_{\mathfrak{V}_1}\circ\rho_w$, for all $w\in \mathfrak{W}$. In fact,
		$$
		\widehat{\rho}_w\circ \tau_{\mathfrak{V}_1} (v)=\tau_{\mathfrak{V}_1}([w,v])=\tau_{\mathfrak{V}_1}\circ \rho_w(v),\qquad \forall w\in \mathfrak{W},\forall v\in\mathfrak{V}\Longrightarrow  \widehat \rho_w\circ \tau_{\mathfrak{V}_1}=\tau_{\mathfrak{V}_1}\circ \rho_w,\quad \forall w\in \mathfrak{W}.
		$$
		Using that $\rho$ and  $\widehat \rho$ are equivariant relative to $\tau_{\mathfrak{V}_1}$, we are going to finally prove that if $v(t)$ takes values in $ \mathfrak{V}$ and $\tau_{\mathfrak{V}_1}(v(t))\neq 0$, then $\tau_{\mathfrak{V}_1}(g(t)_\di v(t))\neq 0$ for every  $g(t)$ in $\mathcal{G}_\mathfrak{W}$. 
		Since  $\dim \,\mathfrak{V}/\mathfrak{V}_1=1$, there exists a unique function $\theta\in \mathfrak{W}^*$ such that $\widehat \rho_w(\tau_{\mathfrak{V}_1}(v))=\theta(w)\tau_{\mathfrak{V}_1}(v)$. 
		
		Consider the case of $g(t)=\exp(f(t)w)$ for some $w\in \mathfrak{W}$, a $t$-dependent function $f(t)$, and an element $v(t)$ in $\mathfrak{V}_\mathbb{R}$. It is clear that
		$$
		g(t)_\di v(t)=\frac{dg}{dt}(t) g(t)^{-1}+{\rm Ad}_{g(t)}v(t)=\frac{df}{dt}(t) w+v(t)+f(t)[w,v(t)]+\frac{f^2(t)}{2}[w,[w,v(t)]]+\ldots
		$$
		Applying $\tau_{\mathfrak{V}_1}$ on both sides, using the assumptions on $\mathfrak{W}$ and $\mathfrak{V}$, and recalling the $\widehat \rho$ is equivariant relative to $\tau_{\mathfrak{V}_1}$, we obtain
		$$
		\begin{aligned}
			\tau_{\mathfrak{V}_1}(g(t)_\di v(t))&=\tau_{\mathfrak{V}_1}\left(\frac{dg}{dt}g(t)^{-1}\right)+\tau_{\mathfrak{V}_1} (v(t))+\tau_{\mathfrak{V}_1} \left([f(t)w,v(t)]+\frac {f^2(t)}2[ w,[w,v(t)]]+\ldots\right)\\
			&=(1+f(t)\theta(w)+\frac {f^2(t)}2\theta^2(w)+\ldots )\tau_{\mathfrak{V}_1} (v(t))\\
			&=e^{f(t)\theta(w)}\tau_{\mathfrak{V}_1}(v(t)).
		\end{aligned}
		$$
		Therefore, if $\tau_{\mathfrak{V}_1}(v(t))\neq 0$, then $\tau_{\mathfrak{V}_1}(g(t)_\di v(t))\neq 0$. 
		Since every $g(t)$ in $ \mathcal{G}_\mathfrak{W}$ can be written as a product of elements of the form $g(t)=\prod_{i=1}^s\exp(f_i(t)w_i)$ for some elements $w_1,\ldots,w_s$ of $\mathfrak{W}$ and functions $f_1(t),\ldots,f_s(t)$, it follows from the above result that if $\tau_{\mathfrak{V}_1}(v(t))\neq0$, then $\tau_{\mathfrak{V}_1}(g(t)_\di(v(t)))\neq 0$ for every $g(t)$ in $\mathcal{G}_\mathfrak{W}$.
	\end{proof}
	
	As an immediate consequence of the proof of Proposition \ref{Red2}, we obtain the following corollary. 
	
	\begin{corollary}\label{Corollary} Let us assume the same conditions of Theorem \ref{Red2} and let ${\rm i}\widehat H_1,\ldots,{\rm i}\widehat H_r$ be a basis of $\mathfrak{V}$ such that ${\rm i}H_1,\ldots,{\rm i}H_{r-1}$ is a basis of $\mathfrak{V}_1$. If $\theta$ is the element of $\mathfrak{W}^*$ associated with the induced representation $\widehat\rho:\mathfrak{W}\rightarrow {\rm End}(\mathfrak{V}/\mathfrak{V}_1)$, namely if $\tau_{\mathfrak{V}_1}:\mathfrak{V}\rightarrow \mathfrak{V}/\mathfrak{V}_1\simeq \langle {\rm i}\widehat H_r\rangle$, the $\theta$ is the only element in $\mathfrak{W}^*$ such that $\tau_{\mathfrak{V}_1}([a,{\rm i}\widehat{H}_r])=\theta(a)\tau_{\mathfrak{V}_1}({\rm i}\widehat{H}_r)$ for every $a\in \mathfrak{W}$, then, for any functions $a_1(t),\ldots,a_r(t)$, we have
		$$
		U(t)=\prod_{i=1}^r\exp(a_i(t){\rm i}\widehat H_i)\,\,\Longrightarrow\,\,
		\tau_{\mathfrak{V}_1}(U(t)_\di {\rm i}\widehat{H}(t))=\left[\prod_{i=1}^r\exp(a_i(t)\theta({\rm i}\widehat H_i))\right]{\rm i}\widehat H_{r}.
		$$
	\end{corollary}
	\section{Applications of quantum quasi-Lie schemes}\label{Applications}
	This section concerns the application of the theory of quantum quasi-Lie schemes to the quantum anharmonic oscillator model given by $n$ interacting particles via a $t$-dependent Hermitian operator of the form
	\begin{equation}\label{Ham}
		\widehat{H}_{nH}(t):=\frac{1}{2}\sum_{i=1}^n\left({ \widehat{p}}_i^2+\omega^2(t)\widehat{
			x}_i^2\right)+c(t){U_{nH}}({ \widehat{x}}_1,\ldots,{ \widehat{x}}_n),
	\end{equation}
	where $c(t)$ is a non-vanishing real function, $\omega(t)$ is any real $t$-dependent 
	function describing a sort of $t$-dependent frequency, and ${U_{nH}}({ \widehat{x}}_1,\ldots,{ \widehat{x}}_n)$ is a quantum
	potential determined by an homogeneous polynomial of order $k$ depending on the position operators ${
		\widehat{x}}_1,\ldots,{ \widehat{x}}_n$, i.e.
	${U_{nH}}(\lambda { \widehat{x}}_1,\ldots,\lambda{ \widehat{x}}_n)=\lambda^k {U_{nH}}({ \widehat{x}}_1,\ldots,{ \widehat{x}}_n)$ for every $\lambda\in \mathbb{R}$. As a particular instance, (\ref{Ham}) covers many types of anharmonic quantum oscillators, which have been extensively studied in the literature \cite{LL65,SCKW97,SCWZ99}.
	
	Perelomov found some conditions ensuring the existence of a $t$-dependent change of variables mapping a classical analogue of (\ref{Ham}) onto an autonomous Hamiltonian system, up to a trivial time-reparametrisation, related to the {Hamiltonian}
	\begin{equation}\label{HamAut}
		H_{nH}=\frac{1}{2}\sum_{i=1}^n  p_i^2+{U_{nH}}( x_1,\ldots, x_n),
	\end{equation}
	where $U_{nH}$ is now understood as a real homogeneous polynomial of order $k$ on the position variables $x_1,\ldots,x_n$ (see \cite{Pe78} for details).
	As many Hamiltonians of the form (\ref{HamAut}) are known to be explicitly integrable, the $t$-dependent change of variables found by Perelomov relates the solutions of such Hamiltonians with the solutions of an associated non-autonomous one $H_{nH}$ \cite{Pe78}.

	Perelomov left as an open problem to look for a quantum analogue of his results (cf. \cite{Pe78}). Additionally, the classical anharmonic oscillator related to the $t$-dependent {Hamiltonian} (\ref{Ham}) was briefly {analysed} via quasi-Lie schemes in \cite{CGL08}. Subsequently, these results are extended to the quantum realm by means of a quantum quasi-Lie scheme, resulting in a solution to Perelomov's open problem.
	
	Let us build up a quantum quasi-Lie scheme $S(\mathfrak{W}_{nH},\mathfrak{V}_{nH})$ to deal with the $t$-dependent {Hamiltonian} operator $-{\rm i}\widehat{H}_{nH}(t)$. This demands to determine subspaces $\mathfrak{W}_{nH},\mathfrak{V}_{nH}$ of skew-Hermitian operators satisfying (\ref{con}) and such that $-{\rm i}\widehat{H}_c(t)$ takes values in $\mathfrak{V}_{nH}$. The construction of $S(\mathfrak{W}_{nH},\mathfrak{V}_{nH})$ is accomplished in the following lemma.
	
	\begin{lemma}\label{lem1} The spaces $\mathfrak{V}_{nH}:=\langle {\rm i}\widehat{H}_1,{\rm i}\widehat{H}_2,{\rm i}\widehat{H}_3,{\rm i}\widehat{H}_4\rangle$ and $\mathfrak{W}_{nH}:=\langle {\rm i}\widehat{H}_2,{\rm i}\widehat{H}_3\rangle$, with 
		$$
		\widehat{H}_1:=\sum_{i=1}^n\frac{{ \widehat{p}}_i^2}2,\quad \widehat{H}_2:=\frac{1}{4}\sum_{i=1}^n({ \widehat{x}}_i{ \widehat{p}}_i+{ \widehat{p}}_i{ \widehat{x}}_i),\quad \widehat{H}_3:=\sum_{i=1}^n\frac{{ \widehat{x}}_i^2}2,\quad \widehat{H}_4:={U_{nH}}(\widehat { x}_1,\ldots,\widehat { x}_n),
		$$
		give rise to quantum quasi-Lie scheme $S(\mathfrak{W}_{nH},\mathfrak{V}_{nH})$ such that $-{\rm i}\widehat{H}_{nH}(t)$ takes values in $\mathfrak{V}_{nH}$.
	\end{lemma}
	\begin{proof} Let us prove first that $S(\mathfrak{W}_{nH},\mathfrak{V}_{nH})$ is a quantum quasi-Lie scheme. It is immediate that $\mathfrak{V}_{nH}$ and $\mathfrak{W}_{nH}$ are finite-dimensional linear spaces of skew-Hermitian operators. 
		
		Since $[{\rm i}\widehat{H}_2,{\rm i}\widehat{H}_3]={\rm i}\widehat{H}_3$,
		the space $\mathfrak{W}_{nH}$ is a real Lie algebra of skew-Hermitian operators as demanded by the definition of a quantum quasi-Lie scheme. Let us verify the only left property of quantum quasi-Lie schemes, namely $[\mathfrak{W}_{nH},\mathfrak{V}_{nH}]\subset \mathfrak{V}_{nH}$. It is immediate that 
		$$
		[{\rm i}\widehat{H}_2,{\rm i}\widehat{H}_1]=-{\rm i}\widehat{H}_1,\qquad [{\rm i}\widehat{H}_3,{\rm i}\widehat{H}_1]=-2{\rm i}\widehat{H}_2,\qquad [{\rm i}\widehat{H}_3,{\rm i}\widehat{H}_4]=0.
		$$
		The only non-immediate part to corroborate the condition $[\mathfrak{W}_{nH},\mathfrak{V}_{nH}]\subset \mathfrak{V}_{nH}$ is to show that $[{\rm i}\widehat{H}_2,{\rm i}\widehat{H}_4]\subset \mathfrak{V}_{nH}$. As $[{ \widehat{x}}_j,{ \widehat{p}}_j]={\rm i}{ I}$ for $j=1,\ldots,n$, then ${ \widehat{p}}_j{ \widehat{x}}_j={ \widehat{x}}_j{ \widehat{p}}_j-{\rm i}\widehat{ I}$ and
		\begin{equation}\label{homog}
			\sum_{j=1}^n({ \widehat{x}}_j{ \widehat{p}}_j+{ \widehat{p}}_j{ \widehat{x}}_j)=2\sum_{j=1}^n{ \widehat{x}}_j{ \widehat{p}}_j-{\rm i}n\widehat{ I},
		\end{equation}
		where $\widehat{I}$ is the identity operator on $\mathcal{H}$. Consequently,
		$$
		[{\rm i}\widehat{H}_2,{\rm i}\widehat{H}_4]=-\frac 12\left[\sum_{j=1}^n{ \widehat{x}}_j{ \widehat{p}}_j,{U_{nH}}({ \widehat{x}}_1,\ldots,{ \widehat{x}}_n)\right]=\frac {\rm i}2\sum_{j=1}^n\widehat{ x}^j\frac{\partial {U_{nH}}}{\partial x^j}({ \widehat{x}}_1,\ldots,{ \widehat{x}}_n).
		$$
		In view of the {\it Euler's homogeneous function theorem} \cite{Bu78} and recalling that ${U_{nH}}$ is a homogeneous function of degree $k$, it follows that
		$$
		\sum_{j=1}^nx^j\frac{\partial {U_{nH}}}{\partial x^j}=k{U_{nH}} \,\,\Longrightarrow [{\rm i}\widehat{H}_2,{\rm i}\widehat{H}_4]=\frac {\rm i}2k{U_{nH}}({ \widehat{x}}_1,\ldots,{ \widehat{x}}_n).
		$$
		This yields $[{\rm i}\widehat{H}_2,{\rm i}\widehat{H}_4]=\frac{{\rm i}}2k\widehat{H}_4$ and
			$[\mathfrak{W}_{nH},\mathfrak{V}_{nH}]\subset \mathfrak{V}_{nH}$. Thus, $\mathfrak{W}_{nH},\mathfrak{V}_{nH}$ give rise to a quantum quasi-Lie scheme $S(\mathfrak{W}_{nH},\mathfrak{V}_{nH})$.
			
			Finally, the form of $\mathfrak{V}_{nH}$ and $-{\rm i}\widehat{H}_{nH}(t)$, which is given by (\ref{Ham}), {ensures} that $-{\rm i}\widehat{H}_{nH}(t)$ takes values in $\mathfrak{V}_{nH}$, which finishes the proof.
		\end{proof}
		
		Once it is been stated that $S(\mathfrak{W}_{nH},\mathfrak{V}_{nH})$ is a quantum quasi-Lie scheme and it can be used to describe the $t$-dependent Hamiltonian operators in (\ref{Ham}), it is time to 
		apply the group of $S(\mathfrak{W}_{nH},\mathfrak{V}_{nH})$ to analyse $-{\rm i}\widehat H_{nH}(t)$. In particular, a relevant case occurs when an element $U(t)$ of $\mathcal{G}_{\mathfrak{W}_{nH}}$ allows us to map $-{\rm i}\widehat H_{nH}(t)$ into a $t$-dependent skew-Hermitian operator $-{\rm i}\widehat{H}_{U}'(t):=-U(t)_\di{\rm i}\widehat{H}_{nH}(t)$ taking values in a one-dimensional subspace of $\mathfrak{V}_{nH}$. In such a case, a time-reparametrisation allows us to solve the transformed $t$-dependent Schr\"odinger equation related to $-{\rm i}\widehat{H}_{U}'(t)$ and, applying  $U^{-1}(t)$ to the general solution of the latter, we obtain the general solution of the initial $t$-dependent Schr\"odinger equation.
		
		\begin{lemma}\label{NoCo} The $t$-dependent operator $-U(t)_\di {\rm i}\widehat H_{nH}(t)$ takes values in a one-dimensional subspace of $\mathfrak{V}_{nH}$ if and only if $U(t)$, which is assumed without loss of generality to take the form 
			\begin{equation}\label{decom}
				U(t)=\exp\left({\rm i}\alpha(t)\widehat{H}_2\right)\exp\left({\rm i}\beta(t)\widehat{H}_3\right),
			\end{equation}
			is such that {$c(t)e^{\frac{\alpha(t)(k+2)}2}$} is a non-zero constant.
		\end{lemma}
		\begin{proof}
			In view of the Lie algebra representation of $S(\mathfrak{W}_{nH},\mathfrak{V}_{nH})$, namely $\rho:w\in \mathfrak{W}_{nH}\mapsto [w,\cdot]\in {\rm End}( \mathfrak{V}_{nH})$,  the quantum quasi-Lie scheme $S(\mathfrak{W}_{nH},\mathfrak{V}_{nH})$ admits two subrepresentations given by $\mathfrak{V}_1:=\langle {\rm i}\widehat H_1,{\rm i}\widehat H_2,{\rm i}\widehat H_3\rangle$ and $\mathfrak{V}_2:=\langle{\rm i}\widehat H_2,{\rm i}\widehat H_3,{\rm i}\widehat H_4 \rangle$. 
			
			Since $c(t)\neq 0$, the $t$-dependent operator $-{\rm i}\widehat H_{nH}(t)$ has no zero projection neither onto $\mathfrak{V}/\mathfrak{V}_1$ nor onto $\mathfrak{V}/\mathfrak{V}_2$. Then, Theorem \ref{Red2} yields that the $t$-dependent skew-Hermitian  operator of the form $-U(t)_\di{\rm i}\widehat{H}_{nH}(t)$ will never take values neither in $\mathfrak{V}_1$ nor {in} $\mathfrak{V}_2$ for any $U(t)$ in $ \mathcal{G}_{\mathfrak{W}}$. In other words, if we write $-U(t)_\di {\rm i}\widehat H_{nH}(t)$ in the  basis ${\rm i}H_1,\ldots,{\rm i}H_4$ of $\mathfrak{V}_{nH}$, then the $t$-dependent coefficients relative to ${\rm i}H_1$ or ${\rm i}H_4$ will never vanish. Moreover, as we want $-U(t)_\di{\rm i}\widehat H_{nH}(t)$ to take values in a one-dimensional subspace of $\mathfrak{V}_{nH}$, the $t$-dependent coefficients relative to ${\rm i}H_4$ and ${\rm i}H_1$ must be proportional. Let us determine their values exactly without determining the exact form of $-U(t)_\di{\rm i}\widehat H_{nH}(t)$. 
			
			As in the proof of Theorem \ref{Red2}, the representation of $S(\mathfrak{W}_{nH},\mathfrak{V}_{nH})$ induces a representation $\widehat \rho_{\mathfrak{V}_2}:\mathfrak{W}_{nH}\rightarrow {\rm End}(\mathfrak{V}/\mathfrak{V}_2)$ such that $(\widehat\rho_{\mathfrak{V}_2})_{{\rm i}\widehat H_3}({{\rm i}\widehat H_1})=0$ i $(\widehat\rho_{\mathfrak{V}_2})_{{\rm i}\widehat H_2}({{\rm i}\widehat H_1})=-{\rm i}\widehat H_1$. Hence, $\theta\in \mathfrak{W}_{nH}^*$ given by $\theta({\rm i}\widehat H_3)=0$ and $\theta({\rm i}\widehat H_2)=-1$ is the unique element of $\mathfrak{W}_{nH}^*$ such that $\tau_{\mathfrak{V}_2}[\widehat A, \widehat B]=\theta(\widehat A)\tau_{\mathfrak{V}_2}(\widehat B)$ for every $\widehat A\in \mathfrak{W}_{nH}$ and $\widehat B\in \mathfrak{V}_{nH}$. If $U(t)$ has the form (\ref{decom}), then  Corollary \ref{Corollary} gives that 
			$
			\tau_{\mathfrak{V}_2}(U(t)_\di {\rm i}\widehat H(t))=e^{-\alpha(t)}{\rm i}\widehat H_1.
			$
			
			Meanwhile, the induced representation $\widehat \rho_{\mathfrak{V}_1}:\mathfrak{W}_{nH}\rightarrow {\rm End}(\mathfrak{V}/\mathfrak{V}_1)$ by the second subrepresentation $\mathfrak{V}_1$ is such that $(\widehat\rho_{\mathfrak{V}_1})_{{\rm i}\widehat H_3}({{\rm i}\widehat H_4})=0$ i $(\widehat\rho_{\mathfrak{V}_1})_{{\rm i}\widehat H_2}({{\rm i}\widehat H_4})={\rm i}k/2\widehat H_4$. Hence, $\theta$ is determined by $\theta({\rm i}H_3)=0$ and $\theta({\rm i}\widehat H_2)=k/2$. Thus, if $U(t)$ has the form (\ref{decom}), then   Corollary \ref{Corollary} gives that 
			$
			\tau_{\mathfrak{V}_1}(U(t)_\di {\rm i}H(t))=e^{k\alpha(t)/2}{\rm i}\widehat H_4.
			$
			
			Since the $t$-dependent coefficients of $\tau_{\mathfrak{V}_1}(U(t)_\di {\rm i}\widehat H(t))$ and $\tau_{\mathfrak{V}_2}(U(t)_\di {\rm i}\widehat H(t))$ must be proportional, it follows that 
			$
			c(t)e^{\frac{\alpha(t)(k+2)}2}
			$ must be a constant.
		\end{proof}
		
		Lemma \ref{NoCo} provides an easily derivable necessary condition to map $-{\rm i}\widehat H_{nH}(t)$ into a new $t$-dependent operator $-{\rm i}\widehat H'_{nH}(t)$ taking values in a one-dimensional subspace of $\mathfrak{V}_{nH}$ via an $U(t)$ belonging to $\mathcal{G}_{\mathfrak{W}_{nH}}$. Nevertheless, the determination of sufficient conditions to obtain such a $-{\rm i}\widehat H_{nH}'(t)$ demands to determine the action of a generic $U(t)$   on $-{\rm i}\widehat H_{nH}(t)$. If $U(t)$ is in $\mathcal{G}_{\mathfrak{W}_{nH}}$, then
		$$
		-{\rm i}\widehat{H}'_{nH}(t)= U(t)_\di (-{\rm i}\widehat{H}_{nH}(t))=\frac{dU}{dt}(t)U(t)^{-1}+{\rm Ad}(U(t))(-{\rm i}\widehat{H}_{nH}(t)),
		$$
		and Theorem \ref{MainT} ensures that $U(t)_{\di}(-{\rm i}\widehat{H}_{nH}(t))$ takes values in $\mathfrak{V}_{nH}$. Every element $U(t)$ in $ \mathcal{G}_{\mathfrak{W}_{nH}}$ can be written in the form
		\begin{equation}\label{UniTrans}
			U(t):=\exp\left({\rm i}\alpha(t) \widehat{H}_2\right)\exp\left({\rm i}\beta(t) \widehat{H}_3\right).
		\end{equation}
		To obtain $-{\rm i}\widehat{H}'_{nH}=U(t)_{\di}(-{\rm i}\widehat{H}_{nH}(t))$, it is necessary to use that
		{\small
			$$
			\begin{aligned}
				{\rm Ad}\left[\exp\left({\rm i}\lambda \widehat{H}_2\right)\right]\widehat{H}_1&=e^{-\lambda}\widehat{H}_1,\qquad &{\rm Ad}\left[\exp\left({\rm i}\lambda \widehat{H}_3\right)\right]\widehat{H}_1&=\widehat{H}_1-2\lambda\widehat{H}_2+\lambda^2\widehat{H}_3,\\
				{\rm Ad}\left[\exp\left({\rm i}\lambda \widehat{H}_2\right)\right]\widehat{H}_3&=e^{\lambda}\widehat{H}_3,\qquad &{\rm Ad}\left[\exp\left({\rm i}\lambda \widehat{H}_3\right)\right]\widehat{H}_2&=\widehat{H}_2-\lambda \widehat{H}_3,\\
				{\rm Ad}\left[\exp\left({\rm i}\lambda \widehat{H}_2\right)\right]\widehat{H}_4&={ e^{\frac{k\lambda}2}\widehat{H}_4},\qquad &{\rm Ad}\left[\exp\left({\rm i}\lambda \widehat{H}_3\right)\right]\widehat{H}_4&=\widehat{H}_4.
			\end{aligned}
			$$}
		Then,
		\begin{equation}\label{transformed}
			-\widehat{H}'_{nH}(t)=-e^{-\alpha(t)}\widehat{H}_1+\left(\frac{d\alpha}{dt}(t)+2\beta(t)\right)\widehat{H}_2+e^{\alpha(t)}\left(\frac{d
				\beta}{dt}(t)-\beta^2(t)-\omega^2(t)\right)\widehat{H}_3-e^{\frac{k \alpha(t)}2} c(t)\widehat{H}_4,
		\end{equation}
		where we recall that $\omega(t)$ is the $t$-dependent frequency of our initial $t$-dependent Hermitian operator (\ref{Ham}). 
		
		Note also that Proposition \ref{Red1} ensures that the group of $S(\mathfrak{W}_{nH},\mathfrak{V}_{nH})$ allows us to map $-{\rm i}\widehat H_{nH}(t)$ into a new  $-{\rm i}\widehat{H}'_{nH}(t)$ taking values in $\langle{\rm i}\widehat H_1,{\rm i}\widehat H_4 \rangle$ for some $U(t)$ in $ \mathcal{G}_{\mathfrak{W}_{nH}}$. For instance, suppose that $\beta(t) $ and $\alpha(t)$ are such that
		\begin{equation}\label{Con2}
			\frac{d\beta}{dt}(t)=\beta^2(t)+\omega^2(t),\qquad \frac{d\alpha}{dt}(t)+2\beta(t)=0,
		\end{equation}
		and we recall that $\alpha(0)=\beta(0)=0$.
		Under such conditions, 
		$$
		\widehat{H}'_{nH}(t)=e^{-\alpha(t)}\widehat{H}_1+e^{\frac{\alpha(t) k}2} c(t)\widehat{H}_4.
		$$
		To ensure that $-{\rm i}\widehat{H}'_{nH}(t)$ takes values in a one-dimensional subspace of $\mathfrak{V}_{nH}$, we have to recall the condition in (\ref{NoCo}), i.e. 
		\begin{equation}\label{IntCon}
			c(t)e^{\frac{k\alpha(t)}{2}}=e^{-\alpha(t)}l\Longrightarrow
			c(t)e^{\frac{\alpha(t)(k+2)}{2}}=l,
		\end{equation}
		for a non-cero constant $l\in \mathbb{R}$. Note that we can assume without loss of generality that $l=1$ by redefining $\widehat H_4$. 
		
		Using (\ref{IntCon}) in (\ref{Con2}), we obtain that if $k\neq 2$ and  there exist $\alpha(t)$ and $\beta(t)$ satisfying the previous conditions, then 
		\begin{equation}\label{FinCon}
			(k+2)\frac{d^2c}{dt^2}(t)=(k+3)\frac{(dc/dt)^2(t)}{c(t)}+\omega^2(t)(k+2)^2c(t).
		\end{equation}

		Conversely, if (\ref{FinCon}) holds, then there exist $\alpha(t)$ and $\beta(t)$ satisfying (\ref{Con2}) and (\ref{IntCon}).
		
		It is worth noting that the definition of a new variable $v_c:=dv/dt$ allows us to transform (\ref{FinCon}) into a Lie system related to a VG Lie algebra 
		isomorphic to $\mathfrak{sl}(2,\mathbb{R})$. In fact, if $k=0$, this is indeed a type of Kummer--Schwarz equation of second-order studied in \cite{CGL12}.
		
		Assume that (\ref{FinCon}) is satisfied. Then, the $t$-dependent
		Sch\"{o}dinger equation related to $\widehat{H}'_{nH}(t)$ reads
		$$
		{\rm i}\frac{\partial \psi}{\partial t}=\widehat{H}'_{nH}(t)\psi=e^{-\alpha(t)}(\widehat{H}_1+\widehat{H}_4)\psi.
		$$
		and by means of the time-reparametrization
		$$
		\tau(t):=\int^t e^{-\alpha(t')}dt'=\int^t[c(t') l^{-1}]^{\frac{2}{k+2}}dt',
		$$
		the previous $t$-dependent  Schr\"{o}dinger equation can be mapped into
		$$
		\frac{\partial \psi'}{\partial \tau}=-{\rm i}(\widehat{H}_1+\widehat{H}_4)\psi',
		$$
		whose solution is
		$$
		\psi'(t)=\exp\left(-{\rm i}\tau(t)(\widehat{H}_1+\widehat{H}_4)\right)\psi'(0),
		$$
		and $\psi'$ is the transformed solution and $\psi'(0)$ is an arbitrary element of $\mathcal{L}^2(\mathbb{R}^n)$. 
		Hence, the solution of the initial Schr\"odinger equation determined by the $t$-dependent {Hamiltonian} operator $\widehat{H}_{nH}(t)$ can be obtained by inverting the unitary transformation (\ref{UniTrans}), namely
		\begin{equation}\label{solution}
			\psi(t)=\exp\left(-{\rm i}\beta(t)\widehat{H}_3\right)\exp\left(-{\rm i}\alpha(t)\widehat{H}_2\right)\exp\left(-{\rm i}\tau(t)(\widehat{H}_1+\widehat{H}_4)\right)\psi'(0).
		\end{equation}
		Therefore, we have mapped a non-autonomous Schr\"{o}dinger equation into a new one determined by a $t$-independent Hermitian {Hamiltonian} operator, similarly to what it was done by Perelomov in \cite{Pe78}, but in a quantum mechanical way.  Our result is  summarised in the proposition below.
		
		\begin{proposition}\label{hom} Every $t$-dependent Schr\"odinger equation related to a $t$-dependent Hermitian Hamiltonian operator $\widehat H_{nH}(t)$ of the form (\ref{Ham}), with a homogeneous potential of order $k\neq -2$ and a non-vanishing  function $c(t)$ that is a particular solution of  (\ref{FinCon}),
			has a general solution (\ref{solution}), where
			$$
			\alpha(t)=-\frac{2}{k+2}\log c(t),\qquad \beta(t)=\frac{dc/dt}{c(t)(k+2)},\quad \tau(t)=\int^tc(t')^{\frac{2}{k+2}}dt'.
			$$
		\end{proposition}

		\begin{note} It was already noted that to transform $-{\rm i}\widehat H_{nH}(t)$ into an autonomous system up to a time-reparametrization via the group of $S(\mathfrak{W}_{nH},\mathfrak{V}_{nH})$, the condition (\ref{IntCon}) is necessary. Despite that, the condition (\ref{Con2}) is not mandatory and other alternative ones could be developed. This would lead to new integration methods  for (\ref{Ham}).
		\end{note}

		\section{Homogeneous potentials of degree minus two}\label{minos2}
		Proposition \ref{hom} cannot be applied to $t$-dependent Hamiltonian operators of the form (\ref{Ham}) with a homogeneous potential of degree minus two. These potentials include the relevant potential
		\begin{equation*}
			{U}_{nH}({ \widehat{x}_1},\ldots, { \widehat{x}}_n):=\sum_{j<k}\frac{g}{({ \widehat{x}}_j-{ \widehat{x}}_k)^2},\qquad g\in \mathbb{R}\backslash\{0\},
		\end{equation*}
		of a fluid in a $t$-dependent homogeneous trapping potential \cite{Su97}, which is also a type of Calogero-Moser potential, or the celebrated {Smorodinsky--Winternitz} potentials
		\begin{equation*}
			{U}_{nH}({ \widehat{x}_1},\ldots, { \widehat{x}}_n)=\sum_{j=1}^n\frac{g_j}{{ \widehat{x}}_j^2},\qquad g_j\in \mathbb{R},\qquad \sum_{j=1}^ng_j^2\neq 0.
		\end{equation*}
		Nevertheless, the procedure to prove Proposition \ref{hom} can be modified to deal with these pathological potentials with $k=-2$. This is the main aim of this section.

		Lemma \ref{lem1}, Lemma \ref{NoCo}, and the expression for $-{\rm i} \widehat{H}'_{nH}(t)=U(t)_\di (-{\rm i}\widehat H_{nH}(t))$ given by (\ref{transformed}) are still valid when the potential under study is homogeneous of degree minus two. Hence, to map $-{\rm i}\widehat H_{nH}(t)$ into a new $t$-dependent skew-Hermitian  operator $-{\rm i}\widehat{H}'_{nH}(t)=-U(t)_\di{\rm i}H_{nH}(t)$ taking values in a one-dimensional subspace of $\mathfrak{V}_{nH}$, it is mandatory to apply the condition 
		$$
		c(t)e^{\alpha(t)(k+2)/2}=l,
		$$
		for a certain non-zero constant $l\in \mathbb{R}$. Since $k=-2$, the function $c(t)$ becomes a constant $c$. If we assume that $\alpha(t)$ and $\beta(t)$ are particular solutions to (\ref{Con2}), the transformed $t$-dependent Hamiltonian operator (\ref{transformed}) takes the form 
		$
		\widehat H'_{nH}(t)=e^{-\alpha(t)}(\widehat{H}_1+c \widehat{H}_4).
		$ Hence, the general solution to the $t$-dependent Schr\"odinger equation related to $\widehat H'_{nH}(t)$, whose explicit $t$-dependence can be removed after a $t$-reparametrisation, reads
		$$
		\psi'(x,t)=\exp\left(\int^t_0e^{-\alpha(t')}dt'(\widehat{H}_1+\widehat{H}_4)\right)\psi'(x),
		$$
		where $\psi'(x)$ is an arbitrary element of $\mathcal{L}^2(\mathbb{R}^n)$ and $\psi'(x,0)=\psi(x)$. Hence, the general solution to (\ref{Ham}) for a homogeneous potential of degree minus two reads as (\ref{solution}) for some solutions $\alpha(t),\beta(t)$ to equations  to (\ref{Con2}), (\ref{IntCon}) and (\ref{FinCon}). This result is summarised in the following proposition.
		
		\begin{proposition}\label{hom1} Every $t$-dependent Hamiltonian operator $\widehat H(t)$ of the form (\ref{Ham}) with a homogeneous potential of order minus two and $c(t)=c\in \mathbb{R}$
			has a general solution (\ref{solution}) where $\beta(t)$ is a particular solution to the  {Riccati} differential equation
			$
			d\beta/dt=\beta^2+\omega^2(t)
			$
			and
			$$
			\alpha(t)=-2\int^t\beta(t')dt', \quad c_0\in \mathbb{R}\backslash\{0\}, \quad \tau(t)=\int^te^{-\alpha(t')}dt'.$$
		\end{proposition}
		
		Let us apply Proposition \ref{hom1} to the $t$-dependent
		{Hamiltonian} \begin{equation*}
			\widehat H_{nH}(t)=\frac{1}{2}\sum_{j=1}^n\left({ \widehat{p}}_j^2+K(t)\widehat{  x}_j^2\right)+\sum_{j<k}\frac{\lambda(\lambda-1)}{({ \widehat{x}}_j-{ \widehat{x}}_k)^2},
		\end{equation*}
		whose potential is homogeneous of degree minus two. This $t$-dependent {Hamiltonian} operator was {analysed} in \cite{Su97} to study a quantum one-dimensional fluid in a Paul trap \cite{Paul}. Sutherland provided an Ansatz \cite{Su97} to get a particular solution. Here, we recover some of his results from an algorithmic point of view and show that more useful solutions can be derived to study the properties of the quantum system.
		
		Suppose that we take a $t$-dependent transformation of $\mathcal{G}_\mathfrak{W}$ with functions $\beta(t)$  and $\alpha(t)$ satisfying the differential equations
		\begin{equation}\label{fuc}
			\frac {d\beta}{dt}(t)=\beta^2(t)+K(t),\qquad 
			\frac{d\alpha}{dt}(t)+2\beta(t)=0.
		\end{equation}
		In view of Proposition \ref{hom1}, the $t$-dependent Schr\"odinger equation related to $\widehat H_{nH}(t)$  reads
		$$
		\psi(x,t)=\exp(-\beta(t){\rm i}\widehat H_3)\exp(-\alpha(t){\rm i}\widehat H_2)\exp\left(\int^t_0e^{-\alpha(t')}dt'(\widehat{H}_1+\widehat{H}_4)\right)\psi'(x,0).
		$$
		In particular, if we assume 
		$$
		\psi(x,0)=\prod_{j<k}(x_j-x_k)^\lambda,
		$$
		we obtain 
		$$
		(\widehat{H}_1+\widehat{H}_4)\prod_{j<k}(x_j-x_k)^\lambda=0
		$$
		and hence
		$$
		\psi(x,t)=\exp(-\beta(t){\rm i}\widehat H_3)\exp(-\alpha(t){\rm i}\widehat H_2)\prod_{j<k}(x_j-x_k)^\lambda.
		$$
		Since $\prod_{j<k}(x_j-x_k)^\lambda$ is a homogeneous function with degree $\lambda n(n-1)/2$ and in view of (\ref{homog}), it follows that
		$$
		{\rm i}\widehat{H}_2\prod_{j<k}(x_j-x_k)^\lambda=\left(\frac{n}{4}+\frac{\lambda n(n-1)}4\right)\prod_{j<k}(x_j-x_k)^\lambda,
		$$
		and
		$$
		\psi(x,t)=\exp(-\beta(t){\rm i}\widehat H_3)\exp\left(-\alpha(t)n\frac{(1+\lambda(n-1))}{4}\right)\prod_{j<k}(x_j-x_k)^\lambda.
		$$
		In view of the expression of $\widehat H_3$, it follows that 
		$$
		\psi(x,t)=\exp\left(-{\rm i}\,\beta(t)\frac{\sum_{j=1}^nx_j^2}{2}\right)\exp\left(-\alpha(t)n\frac{(1+\lambda(n-1))}{4}\right)\prod_{j<k}(x_j-x_k)^\lambda.
		$$
		From equation (\ref{fuc}), it turns out that 
		$
		d\phi/dt(t)=\beta(t)\phi(t),
		$
		with $\phi(t):=e^{-\alpha(t)/2}$. Hence, 
		$$
		\psi(x,t)=\exp\left(-{\rm i}\,\beta(t)\frac{\sum_{j=1}^nx_j^2}{2}\right)\phi(t)^{-n\frac{1+\lambda(n-1)}{2}}\prod_{j<k}(x_k-x_j)^\lambda,
		$$
		recovering, in this way, the particular solution given in \cite{Su97}. Nevertheless, our approach is more general and our solution is not retrieved through a particular Ansatz
		as in previous literature.

		\section{Conclusions and outlook}\label{Sec:Sum}

		We have proposed the theory of quantum quasi-Lie schemes as a way to transform an initial $t$-dependent Schr\"{o}dinger equation into another one, which can be described by means of the usual 
		theory of quantum Lie systems. This enables us to investigate a larger set of $t$-dependent Schr\"{o}dinger equations than just those related to quantum Lie systems. We have also shown that the theory of quasi-Lie systems admits an equivalent  {generalisation} to that of the theory of Lie systems to the quantum framework.
		
		As a particular instance, we have applied our theory to answer a question made in a paper by Perelomov \cite{Pe78} about the possibility of relating a
		quantum $t$-dependent {Hamiltonian} for a nonlinear oscillator to a quantum $t$-independent one of the same type. We have also obtained the solution of a quantum one-dimensional fluid in a Paul trap.
		
		\section*{Acknowledgements}
		Research of J. de Lucas founded by the project under the contract number 291/2017.
		Partial financial support by research projects  {  PGC2018-098265-B-C31 (MICINN)  and   E38/17R (DGA)}
		are acknowledged. C. Sardón acknowledges project ``Teoría de aproximación constructiva y aplicaciones'' (TACA-ETSII). 
		\qquad
		
	\end{document}